\documentclass[]{llncs}

\usepackage{graphicx}
\usepackage{amsmath}
\usepackage{amssymb}
\usepackage{complexity}
\usepackage{soul}
\usepackage{tcolorbox}
\usepackage{verbatim,tikz}
\usetikzlibrary{calc, shapes}
\usetikzlibrary{patterns}
\usetikzlibrary{decorations.pathmorphing}
\usetikzlibrary{decorations.pathreplacing}
\usepackage{multirow}
\usepackage{threeparttable}
\newcommand{\zigzag}{zigzag}
\usepackage{epstopdf} 
\usepackage{listings}
\usepackage[ruled,vlined]{algorithm2e}
\usepackage{todonotes,hyperref}

\title{Parameterized Complexity of Minimum Membership Dominating Set} %TODO Please add

\titlerunning{Parameterized Complexity of Minimum Membership Dominating Set} %TODO optional, please use if title is longer than one line
\author{Akanksha Agrawal\inst{1} \and Pratibha Choudhary %\footnote{Part of the work was done when the author was a postdoctoral fellow at IIT Madras.}
\inst{2} \and N. S. Narayanaswamy\inst{1} \and K. K. Nisha\inst{1} \and Vijayaragunathan Ramamoorthi\inst{1}}

\institute{Department of Computer Science and Engineering, IIT Madras, Chennai, India\\
\email{\{akanksha,swamy,kknisha,vijayr\}@cse.iitm.ac.in} \and Faculty of Informatics, Czech Technical University in Prague, Czech Republic\\
\email{pratibha.choudhary@fit.cvut.cz}
}

\authorrunning{Agrawal et al.}

\newcommand{\bli}{\textbf{i}}

\newcommand{\ch}{c}
\renewcommand{\D}{\mathcal{D}}

\newcommand{\mcc}{\textsc{Multi-Colored Clique} }
\newcommand{\mcis}{\textsc{Multi-Colored Independent Set} }

\newcommand{\srdsws}{$[\sigma,\rho]$-\textsc{Dominating Set}}
\newcommand{\jlds}{$[j,\ell]$-\textsc{Dominating Set}}
\newcommand{\rjds}{$[1,j]$-\textsc{Dominating Set} }

\newcommand{\T}{\mathcal{T}}

\newcommand{\X}{X}
\renewcommand{\Xi}{X_\bli}

\newcommand{\tw}{\textbf{tw}}
\newcommand{\pw}{\textbf{pw}}

\newcommand{\mmds}{\textsc{MMDS}\xspace}

\newtheorem{observation}[theorem]{\textbf{Observation}}
\newtheorem{claimnew}[theorem]{Claim}

\begin{document}

\maketitle 

\begin{abstract}
Given a graph $G=(V,E)$ and an integer $k$, the {\sc Minimum Membership Dominating Set} (\mmds) problem seeks to find a dominating set $S \subseteq V$ of $G$ such that for each $v \in V$, $|N[v] \cap S|$ is at most $k$. %This problem was first studied  by Narayanaswamy et al., where it was proven to be para $\NP$-hard with respect to the size of membership and \textsc{W}[1]-hard when parameterized by pathwidth of the graph. 
We investigate the parameterized complexity of the problem and obtain the following results about {\sc MMDS}: 
\begin{enumerate}
\item W[1]-hardness of the problem parameterized by the pathwidth (and thus, treewidth) of the input graph. 
\item \textsc{W}[1]-hardness parameterized by $k$ on split graphs.
\item An algorithm running in time $2^{\mathcal{O}(\textbf{vc})} |V|^{\mathcal{O}(1)}$, where $\textbf{vc}$ is the size of a minimum-sized vertex cover of the input graph.
\item An ETH-based lower bound showing that the algorithm mentioned in the previous item is optimal. 
\end{enumerate}
%Given a graph $G=(V,E)$ and a positive integer $k$, the \emph{Minimum Membership Dominating Set} (MMDS) problem seeks to decide if there exists a dominating set $S \subseteq V$ of $G$ such that for each $v \in V$, $|N[v] \cap S| \leq  k$ where $N[v]$ is the closed neighborhood of $v$ in $G$. 
%We show that the MMDS problem is W[1]-hard for the parameter pathwidth by a parameter preserving reduction from \mcc problem. For split graphs, we prove that \mmds is \textsc{W}[1]-hard when parameterized by the size of membership. We also study structural parameterization of \mmds, where we analyze the complexity of  \textsc{MMDS} when parameterized by the \emph{size of vertex cover}. We prove that \textsc{MMDS} admits a fixed-parameter tractable (FPT) algorithm in split graphs, when parameterized by the size of vertex cover. Further, we show that assuming ETH, there does not exist a sub-exponential algorithm for \textsc{MMDS} when parameterized by the size of vertex cover. 
\end{abstract}
% !TEX root = MMDSspecial.tex
\section{Introduction}

For a graph $G = (V,E)$, a set $S \subseteq V$ is a {\em dominating set} for $G$, if for each $v \in V$, either $v \in S$, or a neighbor of $v$ in $G$ is in $S$. The \textsc{Dominating Set} problem takes as input a graph $G=(V,E)$ and an integer $k$, and the objective is to test if there is a dominating set of size at most $k$ in $G$. The \textsc{Dominating Set} problem is a classical NP-hard problem~\cite{GareyJohnson79}, which together with its variants, is a well-studied problem in Computer Science. It is also known under standard complexity theoretic assumption that, \textsc{Dominating Set} cannot admit any algorithm running in time $f(k)\cdot|V|^{\mathcal{O}(1)}$ time, where $k$ is the size of dominating set.\footnote{More formally, in the framework of parameterized complexity (see Section~\ref{sec:prelim} for definitions), the problem is W[2]-hard, and thus we do not expect any FPT algorithm for the problem, when parameterized by the solution size.} A variant of \textsc{Dominating Set} that is of particular interest to us in this paper, is the one where we have an additional constraint that the number of closed neighbors that a vertex has in a dominating set is bounded by a given integer as input.\footnote{For a vertex $v$ in a graph $G = (V,E)$, the closed neighborhood of $v$ in $G$, $N_G[v]$, is the set $\{u \in V \mid \{a,b\} \in E\} \cup \{v\}$.} As \textsc{Dominating Set} is a notoriously hard problem in itself, so naturally, the above condition does not make the problem any easier. The above variant has been studied in the literature, and several hardness results are known for it~\cite{MeybodiFMP20}. Inspired by such negative results, in this paper, we remove the size requirement of the dominating set that we are seeking, and attempt to study the complexity variation for such a simplification. We call this version (to be formally defined shortly) of the \textsc{Dominating Set} problem as \textsc{Minimum Membership Dominating Set }(\textsc{MMDS}, for short). For a graph $G =(V,E)$, a vertex $u \in V$ and a set $S \subseteq V$, the {\em membership} of $u$ in $S$ is $M(u, S) = |N[u] \cap S|$. Next we formally define the  \textsc{MMDS} problem.
\begin{tcolorbox}
\textsc{Minimum Membership Dominating Set (MMDS)}\\
\textbf{Input}: A graph $G=(V, E)$ and a positive integer $k$.\\
\textbf{Parameter:} $k$.\\
\textbf{Question}: Does there exist a dominating set $S$ of $G$ such that $\max_{u\in V}M(u, S) \leq k$?
\end{tcolorbox}
%We use $k$-\textit{mds} to 
We refer to a solution of MMDS as a $k$-membership dominating set ($k$-mds). Unless, otherwise specified, for MMDS, by $k$ we mean the membership.
The term ``membership'' is borrowed from a similar version of the \textsc{Set Cover} problem by Kuhn \textit{et al.}~\cite{KuhnRWWZ05}, that was introduced to model reduction in interference among transmitting base stations in cellular networks. 
\vspace{1mm}
\noindent
%\subsection{Our contribution}
{\bf Our results.} We prove that  the \textsc{MMDS} problem is $\NP$-Complete and study the problem in the realm of parameterized complexity. 
\begin{theorem}
\label{thm:planarhard}
The MMDS problem is \NP-complete on planar bipartite graphs for $k=1$. 
\end{theorem}
This shows that the MMDS problem for the parameter $k$ is \texttt{Para}-{\tt NP}-hard, even for planar bipartite graphs.
In other words, for every polynomial time computable function $f$, there is no $O(n^{f(k)})$-time algorithm for the MMDS problem. Further, our reduction also shows that the MMDS restricted to planar bipartite graphs does not have a $(2-\epsilon)$ approximation for any $\epsilon > 0$.\\

\noindent Having proved the $\NP$-Completeness property of \mmds, we study the problem parameterized by the pathwidth and treewidth of the input graph. (Please see Section~\ref{sec:prelim} for formal definitions of treewidth and pathwidth). We note that \textsc{Dominating Set} parameterized by the treewidth admits an algorithm running in time $3^{\sf tw} |V|^{\mathcal{O}(1)}$~\cite{CyganFKLMPPS15}. In contrast to the above, we show that such an algorithm cannot exist for \textsc{MMDS}.
\begin{theorem}
\label{thm:twhard}
\textsc{MMDS} is \rm W[1]-hard when parameterized by the pathwidth of the input graph. 
\end{theorem} 
We note that the pathwidth of a graph is at least as large as its treewidth, and thus the above theorem implies that MMDS parameterized by the treewidth does not admit any FPT algorithm. We prove Theorem~\ref{thm:twhard} by demonstrating an appropriate parameterized reduction from a well-known W[1]-hard problem called \textsc{Multi-Colored Clique} (see~\cite{FELLOWS200953} for its W[1]-hardness).

Next we study \textsc{MMDS} for split graphs, and prove the following theorem.

\begin{theorem}
\label{thm:splitreduction}
\mmds is \textsc{W}[1]-hard on split graphs when parameterized by $k$. 
\end{theorem}

We prove the above theorem by giving a parameterized reduction from \textsc{Multi-Colored Independent Set}, which is known to be W[1]-hard~\cite{FELLOWS200953}. Our reduction is inspired by the known parameterized reduction from \textsc{Multi-Colored Independent Set} to \textsc{Dominating Set}, where we carefully incorporate the membership constraint and remove the size constraint on the dominating set. We would like to note that \textsc{Dominating Set} is known to be W[2]-complete for split graphs~\cite{dom-split}.

Next we study \textsc{MMDS} parameterized by the vertex cover number of the input graph and show that it admits an FPT algorithm.

\begin{theorem}
\label{thm:vcFPT}
\textsc{MMDS} admits an algorithm running in time $2^{\mathcal{O}(\textbf{vc})} |V|^{\mathcal{O}(1)}$, where $\textbf{vc}$ is the size of a minimum-sized vertex cover of the input graph. 
\end{theorem}

We prove the above theorem by exhibiting an algorithm which is obtained by ``guessing'' the portion of the vertex cover that belongs to the solution, and for the remainder of the portion, solving an appropriately created instance of \textsc{Integer Linear Programming}. 

To complement our Theorem~\ref{thm:vcFPT}, we obtain a matching algorithmic lower bound as follows. 

%It is worth noting that \textsc{Dominating Set} is also FPT when parameterized by the size of vertex cover.

\begin{theorem}
\label{thm:vclowerbound}
Assuming ETH, \textsc{MMDS} does not admit an algorithm running in time $2^{o({\textbf{vc}})} |V|^{\mathcal{O}(1)}$, where $\textbf{vc}$ is the size of a minimum-sized vertex cover of the input graph.
%For a graph $G$ with vertex cover $VC(G)$, $|VC(G)|\geq 3$, there does not exist an algorithm running in time $2^{o(|VC(G)|)}poly(n) $ that solves \textsc{MMDS}, unless ETH fails.
\end{theorem}

%In the parameterized \emph{dominating set problem}, the input is a graph $G$ and a positive integer $k$, and the goal is to decide whether there is a dominating set of $G$ of size at most $k$. 

\noindent
{\bf Related works.} Kuhn \textit{et al.}~\cite{KuhnRWWZ05} introduced the ``membership'' variant, in a spirit similar to what we have, for the \textsc{Set Cover} problem, called \textsc{Minimum Membership Set Cover} (\textsc{MMSC}, for short). For the above problem, they obtained several results, including $\NP$-completeness, an $\mathcal{O}(\ln n)$ approximation algorithm, and a matching approximation hardness result. %approximation algorithm with ratio better than $\mathcal{O}(\ln n)$ unless $\NP \subset TIME(n^{O(\log \log n)})$. 
A special case of the \textsc{MMSC} problem is studied in~\cite{DomGNW06} where the collection of sets have \emph{consecutive ones property}. 
In such a set system, the problem is shown to be polynomial-time solvable.  
Narayanaswamy \textit{et al.}~\cite{DhannyaNR18} and recently, Mitchell and Pandit~\cite{MitchellPandit19} have studied the dual of the \textsc{MMSC} problem which is the \emph{Minimum Membership Hitting Set} (\textsc{MMHS}) problem in various geometric settings.

%The \textsc{Dominating Set} problem one of the very well-known W[2]-complete~\cite{DowneyFellows95} problems.%~\cite{CyganFKLMPPS15,
%DowneyFellows99,
%DowneyFellows13,
%FlumGrohe06,
%Niedermeier06}, 
%with a hope that in practice we are mostly interested in graphs with smaller dominating sets.}
%\textcolor{blue}{In that case, an algorithm which is single exponential in solution size can also be considered practically viable.} 
%However, the dominating set problem 
%Besides, with treewidth of the input graph as the parameter, the minimum dominating set problem is FPT~\cite{AlberBFKN02}.  Recently, Chen and Lin~\cite{ChenLin19} proved that any constant-approximation for the parameterized dominating set problem is W[1]-hard. 
%\textcolor{blue}{
%MMDS seeks a solution without restriction on the size but on the membership value. 
%Whereas, most of the variants of the dominating set problem are seeking a minimum size solution that satisfies some additional properties. 
%}
%There are many results in the literature which consider a combination of membership and the size of the dominating set, which we outline below.

%\noindent {\bf Perfect code.} \noindent 
The problem \textsc{Perfect Code} is a variant of \textsc{Dominating Set} where (in addition to the size constraint) we require the membership of each vertex in the dominating set to be exactly one. \textsc{Perfect Code} is another well-studied variant of \textsc{Dominating Set}, see for instance~\cite{Biggs73,FellowsH91,HuangXZ18,Kratochvil86,Kratochvil87,KratochvilK88,Mollard11,CESATI2002}. Telle~\cite{Telle94J,Telle94T} studied a variant of \textsc{Dominating Set} where two vectors $\sigma,\rho$ are additionally given as input, and the membership of vertices in the dominating set and outside this set needs to be determined by $\sigma$ and $\rho$, respectively.
%\todo{Verify that they are vectors.} 
They obtained several results with respect to parameterized complexity of the above variant of \textsc{Dominating Set}. Also, Chapelle~\cite{Chapelle10} studied the above variant with respect to treewidth as the parameter and gave an algorithm running in time $k^{\tw} |V|^{\mathcal{O}(1)}$, where $\tw$ is the treewidth of the input graph. \textsc{MMDS} with membership constraint $k$, is the same as \srdsws, when $\sigma = [0,k-1]$ and $\rho=[1,k]$, thus the problem also admits such an algorithm.% running in time $k^{\tw} |V|^{\mathcal{O}(1)}$.

%vertices is -\textsc{Dominating Set} problem, which was introduced by Telle~\cite{Telle94J,Telle94T}. 
 
%Telle~\cite{Telle94J,Telle94T} introduced a generalized dominating set called [$\sigma,\rho$]-\textsc{Sominating Set}. Let $\sigma,\rho \subseteq \mathbb{N}$ be two sets of non-negative integers.Given a graph $G=(V,E)$, the \srds problem seeks to decide if there exists a set $S \subseteq V$ such that for each $u \in S$, $|N(u) \cap S| \in \sigma$, and for each $u \notin S$, $|N(v) \cap S| \in \rho$.  Here $N(u)$ refers to the open neighborhood of $u$. Notice that the sets $\sigma$ and $\rho$ are part of the definition of the problem and not part of an input instance. 
%Chapelle~\cite{Chapelle10} showed that the \srds problem is FPT parameterized by the treewidth of the input graph, when the size of the sets $\sigma$ and $\rho$ are bounded by a constant. More precisely, there is a $(|\sigma| + |\rho|)^{\tw(G)}n^{O(1)}$ time algorithm to solve the \srds problem, where $\tw$ is the treewidth of the input graph.

Chellali \textit{et al.}~\cite{ChellaliHHM13} introduced a version called \jlds, where we seek a dominating set where the membership of each vertex is at least $j$ and at most $\ell$. They studied the above problem for the viewpoint of combinatorial bounds on special graph classes like claw-free graphs, $P_4$-free graphs, and caterpillars, for restricted values of $j$ and $\ell$. Recently Meybodi \textit{et al.}~\cite{MeybodiFMP20} studied the problems \rjds and $[1,j]$-\textsc{Total Dominating Set} in the realm of parameterized complexity. Though these problems involve constrained membership, unlike MMDS, they require a membership constraint only on the open neighborhood of vertices. %$[1,j]$-\textsc{Total Dominating Set} and \textsc{Restrained Dominating Set} in a parameterized perspective.  A set $D\subseteq V(G)$ is called a $[1,j]$-Total Dominating Set if every vertex in $V(G)$ has at least one and at most $j$ neighbors in $D$. In $[1,j]$-\textsc{Total Dominating Set} problem, the aim is to check the existence of  a $[1,j]$-Total Dominating Set of size at most $k$, on input of a graph $G$ and a positive integer $k$. Given a graph $G$ and a positive integer $k$, the \textsc{Restrained Dominating Set} problem asks if there is a dominating set $D$ of $G$ of size at most $k$ such that no vertex outside of $D$ has all of its neighbors in $D$.

\section{Preliminaries}
\label{sec:prelim}
We recall in this section some notations and definitions used throughout this article. 
For any two positive integers $x$ and $y$, by $[x,y]$ we mean the set $\{x,x+1,\ldots,y\}$, and by $[x]$ we mean $[1,x]$. 
We assume that all our graphs are simple and undirected. 
Given a graph $G = (V,E)$, $n$ represents the number of vertices, and $m$ represents the number of edges. 
We denote an edge between any two vertices $u$ and $v$ by $uv$. For a subset $S \subseteq V$, by $G[S]$ we mean the subgraph of $G$ induced by $S$, and by $G-S$ we mean $G[V\setminus S]$. 
For every vertex $u \in V$, by $N(u)$ we mean open neighborhood of $u$, and by $N[u]$ we mean closed neighborhood of $u$. 
Similarly, for any set $S \subseteq V$, $N(S) = \bigcup_{u \in S}N(u) \setminus S$ and $N[S] = \bigcup_{u\in S} N[u]$. 
%A {\em planar graph} is a graph that can be drawn in the plane without edge crossings. 
Other than this, we follow the standard graph-theoretic notations based on Diestel~\cite{Diestel12}. 
We refer to the recent books of Cygan \textit{et al.}~\cite{CyganFKLMPPS15} and Downey and Fellows~\cite{DowneyFellows13} for detailed introductions to parameterized complexity.

%Let $U$ be a finite set of elements and $S$ be a collection of subsets of $U$. Then the membership $M(u, S)$ of an element $u$ is defined as $|{T|u\in T, T\in S}|$.

\noindent {\bf Treewidth and pathwidth.} For an undirected graph $G = (V,E)$, a {\em tree decomposition} of $G$ is a pair $(\T, \X)$, where $\T$ is a tree and $\X = \{\Xi \subseteq V \mid \bli \in V(\T)\}$ such that 
\begin{itemize}
\item $\bigcup_{\bli \in V(\T)} \Xi = V$,
\item for each edge $uv \in E$, there exists a node $\bli \in V(\T)$ such that $u,v \in \Xi$, and
\item for each $u \in V$, the set of nodes $\{\bli \in V(\T)\mid u \in \Xi\}$ induces a connected subtree in $\T$. 
\end{itemize}
The {\em width} of a tree decomposition $(\T, \X)$ is $\max_{\bli \in V(\T)} (|\Xi|-1)$. 
The {\em treewidth} of $G$ is the minimum width over all possible tree decompositions of $G$. 
A tree decomposition $(\T, \X)$ is said to be a {\em path decomposition} if $\T$ is a path. 
The {\em pathwidth} of a graph $G$ is the minimum width over all possible path decompositions of $G$. 
Let $\pw(G)$ and $\tw(G)$ denote the pathwidth and treewidth of the graph $G$, respectively. 
The pathwidth of a graph $G$ is one less than the minimum clique number of an interval supergraph $H$ which contains  $G$ as an induced subgraph. It is well-known that the maximal cliques of an interval graph can be linearly ordered such that for each vertex, the maximal cliques containing it occur consecutively in the linear order.  This gives a path decomposition of the interval graph.  
A path decomposition of the graph $G$ is the  path decomposition of the interval supergraph $H$ which contains $G$ as an induced subgraph. In our proofs we start with the path decomposition of an interval graph and then reason about the path decomposition of graphs that are constructed from it.  

%Due to lack of space, proofs of the  starred ($\star$) Lemmas and Claims are moved to Appendix~\ref{apdx:sec:proofs}. 
%\begin{observation}
%\end{observation}
%{Pathwidth A graph $G$ has pathwidth $k$ if it is the edge subgraph of an interval graph with maximum clique size at most $k$, and $k$ is the smallest integer with this property for $G$ \cite{}.}

%\section{Interval graphs}
%\input{input/intervalgreedy.tex}
\section{The MMDS problem on planar bipartite graphs is $\NP$-complete}
\label{sec:npcompleteness}
We show that the MMDS problem is $\NP$-hard  for $k=1$ even when restricted to planar bipartite graphs. The $\NP$-hardness is proved by a reduction from {\sc Planar Positive 1-in-3 SAT} as follows. Let $\phi$ be a boolean formula with no negative literals on $n$ variables $X = \{x_1, x_2,\dots , x_n\}$ having $m$ clauses $C = \{C_1, C_2,\dots , C_m\}$. Further we consider the restricted case when the graph encoding the variable-clause incidence is planar.
Such a boolean formula is naturally associated with a planar bipartite graph $G_\phi = (C \cup X, E)$ where $X = \{x_1, x_2,\dots , x_n\}$, $C = \{C_1, C_2,\dots , C_m\}$  and $E = \{(x_i, C_j ) ~|~\text{variable } x_i$  appears in the clause $C_j\}$.\\

\begin{tcolorbox}

\textbf{PP1in3SAT (Planar Positive 1-in-3 SAT)}\\
\textbf{Input :} A boolean formula $\phi(X)$ without negative literals and  that $G_\phi$ is planar\\
\textbf{Decide:} Does there exist an assignment of values $a_1,a_2,\dots, a_n$ to the variables $x_1, x_2,\dots , x_n$ such that exactly one variable in each clause is set to true?
\end{tcolorbox}

\noindent
It is known that  PP1in3SAT is $\NP$-complete~\cite{MulzerR08}. 
A reduction from  PP1in3SAT to the MMDS problem is shown to prove that the MMDS problem is $\NP$-hard. %Due to lack in space, the hardness reduction is moved to to Appendix~\ref{apdx:sec:proofs}. 
%\begin{theorem}
%The MMDS problem is $\NP$-complete.
%\end{theorem}
\begin{proof} [Proof of Theorem \ref{thm:planarhard} ] 
Given a set $S$, we can check the feasibility of the set $S$ to the instance $(G, k)$ of the MMDS problem in polynomial time. 
Therefore, the MMDS problem is in $\NP$. 
To prove that the MMDS problem is \NP-hard, we present a reduction from  PP1in3SAT.
%$$ PP1in3SAT(\phi)~ \leq_p~~ MMDS(G\phi, 1) $$
\noindent Let $\phi$ be a positive 3-CNF formula such that $G_\phi$ is planar. Now, construct a bipartite graph $\hat G_\phi$ as follows. For each vertex $x_i \in G_\phi$, add an additional vertex $\hat x_i$ and connect this vertex to the corresponding $x_i$ using an edge. Let $\hat X = \{\hat x_1,\hat x_2,\dots,\hat x_n\}$. The resultant graph $\hat G_\phi = \big(((X\cup \hat X)\cup C), E\cup\{(x_i,\hat x_i), 1\leq i\leq n\}\big)$ is also a planar graph. We show that $\phi $ is  satisfiable if and only if $\hat G_\phi$ has a dominating set which hits the closed neighborhood of each vertex exactly once. 
Given a set $S$, we can check the feasibility of the set $S$ to the instance $(G, k)$ of the MMDS problem in polynomial time. 
Therefore, the MMDS problem is in $\NP$. 
To prove that the MMDS problem is \NP-hard, we present a reduction from  PP1in3SAT.
%$$ PP1in3SAT(\phi)~ \leq_p~~ MMDS(G\phi, 1) $$
\noindent Let $\phi$ be a positive 3-CNF formula such that $G_\phi$ is planar. Now, construct a bipartite graph $\hat G_\phi$ as follows. For each vertex $x_i \in G_\phi$, add an additional vertex $\hat x_i$ and connect this vertex to the corresponding $x_i$ using an edge. Let $\hat X = \{\hat x_1,\hat x_2,\dots,\hat x_n\}$. The resultant graph $\hat G_\phi = \big(((X\cup \hat X)\cup C), E\cup\{(x_i,\hat x_i), 1\leq i\leq n\}\big)$ is also a planar graph. We show that $\phi $ is  satisfiable if and only if $\hat G_\phi$ has a dominating set which hits the closed neighborhood of each vertex exactly once. \\

\noindent
Let   $A = \{a_1,a_2,\dots a_n\}$ be a satisfying assignment for $\phi$, such that exactly one variable is set to true in each clause. For the graph $\hat G_\phi$ , construct a set $S \subseteq V$ such as $S = \big\{\{ x_i|a_i = 1, ~a_i\in ~A\}\cup\{\hat x_i~|~a_i=0, ~a_i\in ~A \}, 1\leq i \leq n\big\}$. Clearly, $S$ is a dominating set for $\hat G_\phi$. %For each clause vertex $C_j$, let the three variable vertices be denoted as $x^j_i$. Each $C_j$ is dominated only by $x^j_i = 1$ and each $x^j_i = 0$ is dominated only by the corresponding $\hat x^j_i$.
Consider a clause vertex $c\in C$. Let the three variable vertices adjacent to $c$ be $x, y $ and $z$, out of which only one will be assigned value 1 by the satisfying assignment. Without loss of generality, let this be $y$. The vertex $y$ will dominate $c$ and $\hat y$, and vertices $x$ and $z$ will be dominated by $\hat x$ and $\hat z$ respectively.  Therefore, the $S$ is an MMDS of $\hat G_\phi$ for membership parameter $k=1$.

\noindent  To prove the reverse direction, let $S$ be a dominating set for $\hat G_\phi$ such that the closed neighborhood of each vertex is intersected exactly once by $S$. Since the membership of each vertex in $S$ is 1, it follows that, for each clause $c \in C$, $c \notin S$, and exactly one neighbor of $c$ is in $S$.  Therefore, for each $i$, either $x_i \in S \text{ or } \hat x_i \in S$ %[For the sake of contradiction, assume $C_j\in S$. All three $x^j_i$'s get dominated. However, to dominate $\hat x^j_i$'s either $\hat x^j_i$ or $ x^j_i$ should be there in $S$. In both cases, membership of $x^j_i$ becomes 2].Observe the following properties of $S$
Consider the truth assignment $A$ for $\phi$ as follows: For each $x_i \in S$, assign $a_i = 1$, and for each $x_i \notin S$, assign $a_i = 0$. As there is exactly one vertex $x \in S$ from every $c\in C$, only one literal from every  clause will be satisfied by $A$, and thus $A$ is a satisfying assignment for $\phi$. 
%for {\sc PP1in3SAT} for the CNF formula $\phi$.\\
Hence, the MMDS problem is \NP-complete  for $k=1$ even on planar bipartite graphs.
\end{proof}
\noindent{\bf Remark:}  The reduction also shows that the MMDS problem does not have a polynomial time $(2-\epsilon)$ approximation algorithm unless $\P=\NP$.  This is because such an algorithm can solve the MMDS problem for $k=1$.  Also, we believe that starting with the hardness of planar 3-SAT variants in which each variable occurs exactly 3 times \cite{MulzerR08}, our reduction shows that the MMDS problem on planar bipartite graphs of maximum degree 4 is \NP-complete.
%{\bf State that with respect to membership as the parameter the problem is para-NP hard. What other hardness results can be get by this reduction.  }

\section{W[1]-hardness with respect to pathwidth}
\label{sec:pathwidthhard}
%We show that MMDS problem is W[1]-hard when parameterized by pathwidth by a reduction from \mcc problem. 
%\textcolor{red}
{We prove Theorem~\ref{thm:twhard} by a reduction from the \mcc problem to the MMDS problem.}  
It is well-known that the \mcc problem is  W[1]-hard for the parameter solution size~\cite{DowneyFellows99}.\\
%This shows a fundamental difference between the MMDS and minimum dominating set.  
\begin{tcolorbox}
\mcc\\
\textsf{Input}: A positive integer $k$ and a $k$-colored graph $G$.\\
{\sf Parameter}: $k$\\
{\sf Question}: Does there exists a clique of size $k$ with one vertex from each color class?
\end{tcolorbox}
\noindent
Let $(G=(V,E),k)$ be an instance of the \mcc problem. 
Let $V=(V_1,\ldots, V_k)$ denote the partition of the vertex set $V$. By a partition, we mean the set of all vertices of same color. 
We assume, without loss of generality, $|V_i| = n$ for each $i \in [k]$. We usually use $n$ to denote number of vertices in the input graph. However, we use $n$ here to denote the number of vertices in each color class.
%This is because, in an instance if this property is not met, then we can add isolated vertices to the respective color classes without changing the answer to the question.  
%Let $n$ be the number of vertices in each color class. 
For each $1 \leq i \leq k$, let $V_i = \{u_{i,\ell} \mid 1 \leq \ell \leq n\}$. 
\subsection{Gadget based reduction from \mcc}
For an input instance $(G, k)$ of the \mcc~problem, the reduction outputs an instance $(H, k')$ of the MMDS problem where $k' = n+1$. 
The graph $H$ is constructed using 
%We construct a gadget graphs to represent vertices and edges of the input instance of the \mcc problem.  
%We construct 
two types of gadgets, $\D$ and $I$ (illustrated in Figure~\ref{fig:structures}).
The gadget $I$ is the primary gadget and the gadget $\D$ is secondary gadget that is used to construct the gadget $I$. \\
%The gadget $\D$ is defined as follows. %  \\
\noindent {\bf Gadget of  type $\D$.} 
%Given a pair of vertices $u$ and $v$, the gadget $\D_{u,v}$ is an apex graph which consists of $n+4$ vertices. 
%Let $u$ and $v$ be the apices of $\D_{u,v}$. 
%The vertex set $V(\D_{u,v})\setminus \{u,v\}$ forms an independent set, and each vertex in the set is made adjacent to both $u$ and $v$.  
%Further, $u$ is made adjacent to $v$.
For two vertices $u$ and $v$, the gadget $\D_{u,v}$ is an interval graph consisting of vertices $u$, $v$ and $n+2$ additional vertices that form an independent set. 
The vertices $u$ and $v$ are  adjacent, and both $u$ and $v$ are adjacent to every other vertex. 
We refer to the vertices $u$ and $v$ as {\em heads} of the gadget $\D_{u,v}$. 
Intuitively, for any feasible solution $S$, and for any gadget $\D_{u,v}$, either $u$ or $v$ should be in $S$. Otherwise, remaining $n+2$ vertices must be in $S$ which contradicts the optimality of $S$ because membership for both $u$ and $v$ is at least $n+2$. 
%\noindent {\bf Structure $T$.} Given a vertex $s$, the structure $T_s$ is a 2-level star tree with root $s$. 
%The root $s$ will have $n+1$ children $\ch^1_s,\ldots,\ch_s^{n+1}$. 
%Let $\ch(s)$ denote the set of all children of $s$. 
%For each $1 \leq i \leq n+1$, $\ch_s^i$ will have $n+2$ children. 
%The root $s$ is the only vertex of $T_s$ which will be connected to vertices outside $T_s$. 
%\input{input/figure/structures.tex}
\begin{observation}
\label{obs:DPathwidth}
The pathwidth of the gadget $\D$ is two. Indeed, it is an interval graph with maximum clique of size three and thus, by definition, has pathwidth~2.
\end{observation}
\noindent {\bf Gadget of type $I$.} Let $n \geq 1$ be an integer.
%We begin the construction of the gadget with $2n$ vertices partitioned into two sets where each partition contains $n$ vertices. 
The gadget has two vertices $h_1$ and $h_2$, and two disjoint sets: $A = \{a_1, \ldots, a_n\}$ and $D = \{d_1, \ldots, d_n\}$. 
For each $i\in [n]$, vertices $a_i$ and $d_i$ are connected by the gadget $\D_{a_i, d_i}$. 
Let $h_2$ and $h_1$ be two additional vertices which are adjacent. 
The vertices in the sets $A$ and $D$ are  adjacent to $h_2$ and $h_1$, respectively. 
For each $1 \leq i \leq n$, $a_i$ and $h_1$ are connected by the gadget $\D_{a_i,h_1}$, and
 $d_i$ and $h_2$ are connected by the gadget $\D_{d_i,h_2}$.
%\begin{itemize}
%\item $a_i$ and $h_1$ are connected by the structure $\D_{a_i,h_1}$, and
%\item $d_i$ and $h_2$ are connected by the structure $\D_{d_i,h_2}$.
%\end{itemize}
%The vertex $h_1$ are made adjacent to $h_2$. 
%This completes the construction of the gadget. 
In the reduction a gadget of type $I$ is denoted by the symbol $I$ and an appropriate subscript.  
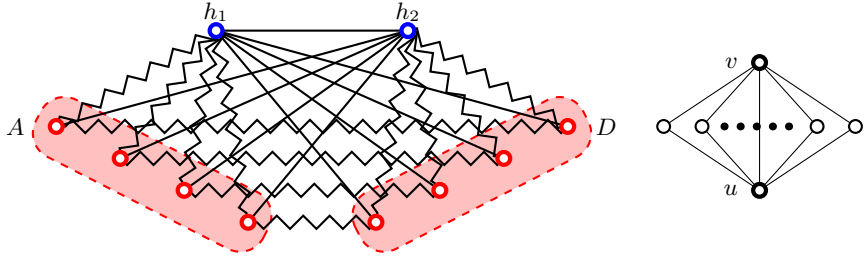
\begin{figure}[h]
\center
\begin{tikzpicture}[scale=0.85]
\coordinate (a1) at (1,4.5);
\coordinate (a2) at (2,4);
\coordinate (a3) at (3,3.5);
\coordinate (a4) at (4,3);

\coordinate (b1) at (1,-0.5);
\coordinate (b2) at (2,0);
\coordinate (b3) at (3,0.5);
\coordinate (b4) at (4,1);

\coordinate (d1) at (6,3);
\coordinate (d2) at (7,3.5);
\coordinate (d3) at (8,4);
\coordinate (d4) at (9,4.5);

\coordinate (c1) at (6,1);
\coordinate (c2) at (7,0.5);
\coordinate (c3) at (8,0);
\coordinate (c4) at (9,-0.5);

\coordinate (e) at (-1,2);
\coordinate (e1) at (-1,0);
\coordinate (f1) at (3.5,-2);
\coordinate (f2) at (6.5,-2);
\coordinate (g) at (11,2);
\coordinate (g1) at (11,0);
\coordinate (h1) at (3.5,6);
\coordinate (h2) at (6.5,6);

%node[ultrathick, black, left=2.05cm, midway] {$A$}
\draw[red, fill=red, fill opacity=0.25, dashed, thick,rounded corners=10pt, rotate around={153:(2.5,3.75))}] (0.5, 4.25) rectangle (4.5,3.25);
%\draw[red, fill=red, fill opacity=0.25, dashed, thick,rounded corners=15pt, rotate around={153:(7.5,0.25))}] (5.5, 0.75) rectangle (9.5,-0.25);
\draw[red, fill=red, fill opacity=0.25, dashed, thick,rounded corners=10pt, rotate around={27:(7.5,3.75))}] (5.5, 4.25) rectangle (9.5,3.25);
%\draw[red, fill=red, fill opacity=0.25, dashed, thick,rounded corners=15pt, rotate around={27:(2.5,0.25))}] (0.5, 0.75) rectangle (4.5,-0.25);

%%%\draw[black, thick] (e) -- (a1);
%%%\draw[black, thick] (e) -- (a2);
%%%\draw[black, thick] (e) -- (a3);
%%%\draw[black, thick] (e) -- (a4);
%%%\draw[black, thick] (e) -- (b1);
%%%\draw[black, thick] (e) -- (b2);
%%%\draw[black, thick] (e) -- (b3);
%%%\draw[black, thick] (e) -- (b4);
%%%
%%%\draw[black, thick] (f1) -- (c1);
%%%\draw[black, thick] (f1) -- (c2);
%%%\draw[black, thick] (f1) -- (c3);
%%%\draw[black, thick] (f1) -- (c4);
%%%\draw[black, thick, decorate, decoration=\zigzag] (f1) -- (b1);
%%%\draw[black, thick, decorate, decoration=\zigzag] (f1) -- (b2);
%%%\draw[black, thick, decorate, decoration=\zigzag] (f1) -- (b3);
%%%\draw[black, thick, decorate, decoration=\zigzag] (f1) -- (b4);
%%%\draw[black, thick, decorate, decoration=\zigzag] (f2) -- (c1);
%%%\draw[black, thick, decorate, decoration=\zigzag] (f2) -- (c2);
%%%\draw[black, thick, decorate, decoration=\zigzag] (f2) -- (c3);
%%%\draw[black, thick, decorate, decoration=\zigzag] (f2) -- (c4);
%%%\draw[black, thick] (f2) -- (b1);
%%%\draw[black, thick] (f2) -- (b2);
%%%\draw[black, thick] (f2) -- (b3);
%%%\draw[black, thick] (f2) -- (b4);
%%%\draw[black, thick] (f2) -- (f1);

\draw[black, thick] (h1) -- (d1);
\draw[black, thick] (h1) -- (d2);
\draw[black, thick] (h1) -- (d3);
\draw[black, thick] (h1) -- (d4);
\draw[black, thick, decorate, decoration=\zigzag] (h1) -- (a1);
\draw[black, thick, decorate, decoration=\zigzag] (h1) -- (a2);
\draw[black, thick, decorate, decoration=\zigzag] (h1) -- (a3);
\draw[black, thick, decorate, decoration=\zigzag] (h1) -- (a4);
\draw[black, thick, decorate, decoration=\zigzag] (h2) -- (d1);
\draw[black, thick, decorate, decoration=\zigzag] (h2) -- (d2);
\draw[black, thick, decorate, decoration=\zigzag] (h2) -- (d3);
\draw[black, thick, decorate, decoration=\zigzag] (h2) -- (d4);
\draw[black, thick] (h2) -- (a1);
\draw[black, thick] (h2) -- (a2);
\draw[black, thick] (h2) -- (a3);
\draw[black, thick] (h2) -- (a4);
\draw[black, thick] (h2) -- (h1);

%%%\draw[black, thick] (g) -- (c1);
%%%\draw[black, thick] (g) -- (c2);
%%%\draw[black, thick] (g) -- (c3);
%%%\draw[black, thick] (g) -- (c4);
%%%\draw[black, thick] (g) -- (d1);
%%%\draw[black, thick] (g) -- (d2);
%%%\draw[black, thick] (g) -- (d3);
%%%\draw[black, thick] (g) -- (d4);

\draw[black, thick, decorate, decoration=\zigzag] (d4) -- (a1);
\draw[black, thick, decorate, decoration=\zigzag] (d3) -- (a2);
\draw[black, thick, decorate, decoration=\zigzag] (d2) -- (a3);
\draw[black, thick, decorate, decoration=\zigzag] (d1) -- (a4);
%%%\draw[black, thick, decorate, decoration=\zigzag] (c1) -- (b4);
%%%\draw[black, thick, decorate, decoration=\zigzag] (c2) -- (b3);
%%%\draw[black, thick, decorate, decoration=\zigzag] (c3) -- (b2);
%%%\draw[black, thick, decorate, decoration=\zigzag] (c4) -- (b1);
%%%
%%%\draw[black, thick, decorate, decoration=\zigzag] (b1) -- (a1);
%%%\draw[black, thick, decorate, decoration=\zigzag] (b2) -- (a2);
%%%\draw[black, thick, decorate, decoration=\zigzag] (b3) -- (a3);
%%%\draw[black, thick, decorate, decoration=\zigzag] (b4) -- (a4);
%%%\draw[black, thick, decorate, decoration=\zigzag] (c1) -- (d1);
%%%\draw[black, thick, decorate, decoration=\zigzag] (c2) -- (d2);
%%%\draw[black, thick, decorate, decoration=\zigzag] (c3) -- (d3);
%%%\draw[black, thick, decorate, decoration=\zigzag] (c4) -- (d4);

%\draw[black, thick] (e) -- (e1);
%\draw[black, thick] (g) -- (g1);
%
%\node[fill, star,star points=5,star point ratio=0.45, blue, thick] at (e1) {.};
%\node[fill, star,star points=5,star point ratio=0.45, blue, thick] at (g1) {.};
%

\draw[red, fill=white, ultra thick] (a1) circle (0.1cm) node[left=0.3cm, black, thick] {$A$};
\draw[red, fill=white, ultra thick] (a2) circle (0.1cm);
\draw[red, fill=white, ultra thick] (a3) circle (0.1cm);
\draw[red, fill=white, ultra thick] (a4) circle (0.1cm);
%%%\draw[red, fill=white, ultra thick] (b1) circle (0.1cm) node[left=0.25cm, black, thick] {$B$};
%%%\draw[red, fill=white, ultra thick] (b2) circle (0.1cm);
%%%\draw[red, fill=white, ultra thick] (b3) circle (0.1cm);
%%%\draw[red, fill=white, ultra thick] (b4) circle (0.1cm);
%%%\draw[red, fill=white, ultra thick] (c1) circle (0.1cm);
%%%\draw[red, fill=white, ultra thick] (c2) circle (0.1cm);
%%%\draw[red, fill=white, ultra thick] (c3) circle (0.1cm);
%%%\draw[red, fill=white, ultra thick] (c4) circle (0.1cm) node[right=0.25cm, black, thick] {$C$};
\draw[red, fill=white, ultra thick] (d1) circle (0.1cm);
\draw[red, fill=white, ultra thick] (d2) circle (0.1cm);
\draw[red, fill=white, ultra thick] (d3) circle (0.1cm);
\draw[red, fill=white, ultra thick] (d4) circle (0.1cm) node[right=0.25cm, black, thick] {$D$};

%%%\draw[blue, fill=white, ultra thick] (e) circle (0.1cm) node[left, black, ultra thick] {$p$};
%%%\draw[blue, fill=white, ultra thick] (f1) circle (0.1cm) node[below, black, thick] {$f_1$};
%%%\draw[blue, fill=white, ultra thick] (f2) circle (0.1cm) node[below, black, thick] {$f_2$};
%%%\draw[blue, fill=white, ultra thick] (g) circle (0.1cm) node[right, black, thick] {$q$};
\draw[blue, fill=white, ultra thick] (h1) circle (0.1cm) node[above, black, thick] {$h_1$};
\draw[blue, fill=white, ultra thick] (h2) circle (0.1cm) node[above, black, thick] {$h_2$};

\makeatletter
\tikzset{
    dot diameter/.store in=\dot@diameter,
    dot diameter=3pt,
    dot spacing/.store in=\dot@spacing,
    dot spacing=10pt,
    dots/.style={
        line width=\dot@diameter,
        line cap=round,
        dash pattern=on 0pt off \dot@spacing
    }
}
\makeatother

\coordinate (a1) at (12,3.5);
\coordinate (a2) at (13.5,4.5);
\coordinate (a3) at (12,5.5);
\coordinate (a4) at (10.5,4.5);

\draw[black] (a1) -- (a2);
\draw[black] (a1) -- (a4);
\draw[black] (a3) -- (a2);
\draw[black] (a3) -- (a4);
\draw[black] (a1) -- ($(a2) + (-0.6,0)$);
\draw[black] (a1) -- ($(a4) + (0.6,0)$);
\draw[black] (a3) -- ($(a2) + (-0.6,0)$);
\draw[black] (a3) -- ($(a4) + (0.6,0)$);
\draw[black] (a1) -- (a3);

\draw[black, fill=white, ultra thick] (a1) circle (0.1cm) node[left=0.15cm] {$u$};
\draw[black, fill=white, thick] (a2) circle (0.1cm);
\draw[black, fill=white, ultra thick] (a3) circle (0.1cm) node[left=0.15cm] {$v$};
\draw[black, fill=white, thick] (a4) circle (0.1cm);
\draw[black, fill=white, thick] ($(a2) + (-0.6,0)$) circle (0.1cm);
\draw[black, fill=white, thick] ($(a4) + (0.6,0)$) circle (0.1cm);

\draw [black, dot diameter=3pt, dot spacing=6pt, dots] ($(a4) + (0.95, 0)$) -- ($(a2) + (-0.9, 0)$);
\end{tikzpicture}
\caption{To the left is the type-$I$  gadget for $n=4$ and to the right is the type-$D$ gadget. The zigzag edges between vertices $u$ and $v$ represent the gadget $\D_{u,v}$. }
\label{fig:structures}
\end{figure}
\begin{claimnew}%$\diamondsuit$\footnote{Proofs of claims and lemmas marked with $\diamondsuit$ can be found in the appendix.}
\label{claim:gagdetPathwidth}
The pathwidth of a gadget type $I$ is at most four. 
\end{claimnew}

\begin{claimproof}
%\label{claim:gagdetPathwidth}
We observe that the removal of the vertices $h_1$ and $h_2$ results in a graph in which for each $i \in [n]$,   there is a connected component consisting $a_i$ and $d_i$ which are the heads of a gadget of type $\D$ and they are both adjacent to $n+2$ vertices of degree 1. Each component is an interval graph with a triangle as the maximum clique and from Observation~\ref{obs:DPathwidth}  is of pathwidth 2.
%$n$ disjoint gadgets of type $\D$  with $n+2$ vertices of degree one ({\bf two}) connected to each head of $\D$. 
%, each $\D$ has pathwidth two. 
%{\bf Further, the added degree one vertices to the heads does not increase the pathwidth. }
Let $(\T', \X')$ be the path decomposition of $I-\{h_1,h_2\}$ with width two. 
Thus adding $h_1$ and $h_2$ into all the bags of the path decomposition $(\T',\X')$ gives a path decomposition for the gadget $I$, and thus the pathwidth of the gadget $I$ is at most 4.
%Since pathwodth of a structure $\D$ is two, and adding $h_1$ and $h_2$ to all the bags in the path decomposition 
%Thus, pathwidth of a gadget is four since
%pathwidth of a structure is two from Observation~\ref{obs:DPathwidth}.
 \end{claimproof}

\noindent
In the following parts, when we refer to a gadget we mean the primary gadget $I$ unless the gadget $\D$ is specified. 
For each vertex and edge in the given graph, our reduction has a corresponding gadget in the instance output by the reduction. \\
{\bf Description of the reduction}.
%An $I$ gadget, by construction forces either $A$ or $D$ from its vertices to the MMDS solution. Moreover, from a vertex partition block and an edge partition block, exactly one set of $A$ vertices could be selected from all of its $n$ many $I$ gadgets. This helps us to convert an \mcc solution to an MMDS solution and vice versa.
%Let $V = V_1 \uplus V_2\uplus \cdots \uplus V_k$. 
For $1 \leq i < j \leq k$, let $E_{i,j}$ denote the set of edges with one end point in $V_i$ and the other in $V_j$, that is $E_{i,j} = \{xy \mid x \in V_i, y\in V_j\}$. \\
For each vertex and edge in $G$, the reduction uses a gadget of type $I$. 
For each $1 \leq i < j \leq k$, the graph $H$ has an induced subgraph $H_i$ corresponding to  $V_i$, and has an induced subgraph $H_{i,j}$ for the edge set $E_{i,j}$.  We refer to $H_i$ as a vertex-partition block and $H_{i,j}$ as an edge-partition block.  Inside  block $H_i$, there is a gadget of type $I$ for each vertex in $V_i$, and
in the block $H_{i,j}$ is a gadget for each edge in $E_{i,j}$.  For a vertex $u_{i,x}$,  $I_x$ denotes the gadget corresponding to $u_{i,x}$ in the partition $V_i$, and for an edge $e$, $I_e$ denotes the gadget corresponding to $e$. 
%The graph $H$ is constructed by {\em blocks}. 
%For each $i \in [k]$, we construct a vertex-partition block $H_i$ in $H$ for the vertex partition $V_i$. 
%For $1 \leq i < j \leq k$, we construct an edge-partition block $H_{i,j}$ in $H$ for the edge partition $E_{i,j}$.  
Finally, the blocks are connected by the connector vertices which we describe below. We next define the structure of a block which we denote by $B$. The definition of the block applies to both the vertex-partition block and the edge-partition block.

\noindent A block $B$ %(both vertex-partition gadget and edge partition block) 
consists of the following gadgets, additional vertices, and edges.
\begin{itemize}
\item The block $B$ corresponding to the vertex-partition block  $H_i$ for any $i \in [k]$ is  as follows: for each $\ell \in [n]$, add a gadget $I_\ell$  to the vertex-partition block $H_i$, to represent the vertex $u_{i,\ell} \in V_i$. 
\item The block $B$ corresponding to the  edge-partition gadget $H_{i,j}$ for any $1 \leq i < j \leq k$ is  as follows: for each $e \in E_{i,j}$, add a gadget $I_e$ in the edge-partition block $H_{i,j}$, to  represent the edge $e$. 
\item In addition to the gadgets, we add $(n+1)(n+3)+2$ vertices to the block $B$  as follows (See Figure~\ref{fig:block} in appendix): 
Let $C(B)$ denote the set $\{f,f',\ch_1,\ch_2,\ldots,$ $\ch_{n+1},b_1,b_2,\ldots,$ $b_{(n+1)(n+2)}\}$, which is the set  of additional vertices that are added to the block $B$. 
Let $C'(B)$ denote the subset  $\{\ch_1,\ch_2,\ldots,\ch_{n+1}\}$. 
For each gadget $I$ in $B$, and for each $t \in [n]$, $a_t$ in $I$ is adjacent to $f$, and the vertex $f$ is  adjacent to $f'$.
Further, the vertex $f'$ is  adjacent to each vertex $\ch_p$ for $p \in [n+1]$. 
Finally, for each $p \in [n+1]$ and $(p-1)(n+2) < q \leq p(n+2)$, $\ch_p$ is  adjacent to $b_q$. 
\end{itemize}

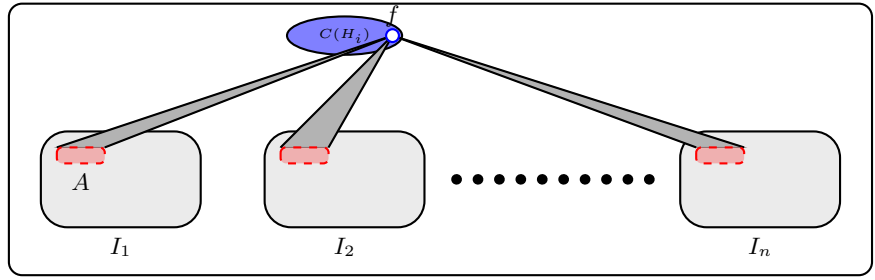
\begin{figure}[h]
\center
\begin{tikzpicture}[scale=0.85]
\makeatletter
\tikzset{
    dot diameter/.store in=\dot@diameter,
    dot diameter=3pt,
    dot spacing/.store in=\dot@spacing,
    dot spacing=10pt,
    dots/.style={
        line width=\dot@diameter,
        line cap=round,
        dash pattern=on 0pt off \dot@spacing
    }
}
\makeatother
\coordinate (x1) at (0,0);
\coordinate (r11) at ($(x1) + (0,2)$);
\coordinate (r21) at ($(x1) + (2.5,0.5)$);
\coordinate (a11) at ($(x1) + (0.25,1.75)$);
\coordinate (a21) at ($(x1) + (1,1.50)$);
\coordinate (x2) at ($(x1) + (3.5,0)$);
\coordinate (r12) at ($(x2) + (0,2)$);
\coordinate (r22) at ($(x2) + (2.5,0.5)$);
\coordinate (a12) at ($(x2) + (0.25,1.75)$);
\coordinate (a22) at ($(x2) + (1,1.50)$);
\coordinate (x3) at ($(x1) + (10,0)$);
\coordinate (r13) at ($(x3) + (0,2)$);
\coordinate (r23) at ($(x3) + (2.5,0.5)$);
\coordinate (a13) at ($(x3) + (0.25,1.75)$);
\coordinate (a23) at ($(x3) + (1,1.50)$);
\coordinate (a31) at ($(a21) + (0,0.25)$);
\coordinate (a32) at ($(a22) + (0,0.25)$);
\coordinate (a33) at ($(a23) + (0,0.25)$);
\coordinate (a) at (5.5,3.5);
%\coordinate (f11) at ($(x1) + (1.5,1.65)$);
%\coordinate (f21) at ($(x1) + (2.15,1.65)$);
%\coordinate (h11) at ($(x1) + (1.5,0.85)$);
%\coordinate (h21) at ($(x1) + (2.15,0.85)$);
%\coordinate (f12) at ($(x2) + (1.5,1.65)$);
%\coordinate (f22) at ($(x2) + (2.15,1.65)$);
%\coordinate (h12) at ($(x2) + (1.5,0.85)$);
%\coordinate (h22) at ($(x2) + (2.15,0.85)$);
%\coordinate (f13) at ($(x3) + (1.5,1.65)$);
%\coordinate (f23) at ($(x3) + (2.15,1.65)$);
%\coordinate (h13) at ($(x3) + (1.5,0.85)$);
%\coordinate (h23) at ($(x3) + (2.15,0.85)$);
%\coordinate (f) at (7,3.5);
%\coordinate (h) at (6.75,-1);

\draw[black, thick, fill=blue!50] ($(a) + (-0.75,0)$) ellipse (0.9cm and 0.3cm) node[black] {\tiny{$C(H_i)$}};

\draw[black, thick, fill=lightgray, fill opacity=0.3, rounded corners=10pt] (r11) rectangle (r21);
\draw[black, thick, fill=lightgray, fill opacity=0.3, rounded corners=10pt] (r12) rectangle (r22);
\draw[black, thick, fill=lightgray, fill opacity=0.3, rounded corners=10pt] (r13) rectangle (r23);

\filldraw[draw=gray, fill=gray!60] (a11) -- (a31) -- (a) -- (a11) -- cycle;
\filldraw[draw=gray, fill=gray!60] (a12) -- (a32) -- (a) -- (a12) -- cycle;
\filldraw[draw=gray, fill=gray!60] (a13) -- (a33) -- (a) -- (a13) -- cycle;

%\draw[black, thick] (f11) -- (f);
%\draw[black, thick] (f21) -- (f);
%\draw[black, thick] (f12) -- (f);
%\draw[black, thick] (f22) -- (f);
%\draw[black, thick] (f13) -- (f);
%\draw[black, thick] (f23) -- (f);
%\draw[black, thick] (h11) -- (h);
%\draw[black, thick] (h21) -- (h);
%\draw[black, thick] (h12) -- (h);
%\draw[black, thick] (h22) -- (h);
%\draw[black, thick] (h13) -- (h);
%\draw[black, thick] (h23) -- (h);
%\draw[black, thick] (f) -- ($(f) + (1.5,0)$);
%\draw[black, thick] (h) -- ($(h) + (1.5,0)$);
\draw[black, thick] (a11) -- (a);
\draw[black, thick] (a31) -- (a);
\draw[black, thick] (a12) -- (a);
\draw[black, thick] (a32) -- (a);
\draw[black, thick] (a13) -- (a);
\draw[black, thick] (a33) -- (a);
%\draw[black, thick] (a) -- ($(a) + (-1.5,0)$);

%\node[fill, star,star points=5,star point ratio=0.45, blue, thick] at ($(a) + (-1.5,0)$) {.};
%\node[fill, star,star points=5,star point ratio=0.45, blue, thick] at ($(f) + (1.5,0)$) {.};
%\node[fill, star,star points=5,star point ratio=0.45, blue, thick] at ($(h) + (1.5,0)$) {.};

\draw[red, dashed, thick, fill=red, fill opacity=0.25, rounded corners=2pt] (a11) rectangle (a21);
\draw[red, dashed, thick, fill=red, fill opacity=0.25, rounded corners=2pt] (a12) rectangle (a22);
\draw[red, dashed, thick, fill=red, fill opacity=0.25, rounded corners=2pt] (a13) rectangle (a23);
%\draw[blue, thick, fill=white] (f11) circle (0.1cm) node[black, below] {$\tiny{h_1}$};
%\draw[blue, thick, fill=white] (f21) circle (0.1cm) node[black, below] {$h_2$};
%\draw[blue, thick, fill=white] (h11) circle (0.1cm) node[black, left] {$f_1$};
%\draw[blue, thick, fill=white] (h21) circle (0.1cm) node[black, left] {$f_2$};
%\draw[blue, thick, fill=white] (f12) circle (0.1cm);
%\draw[blue, thick, fill=white] (f22) circle (0.1cm);
%\draw[blue, thick, fill=white] (h12) circle (0.1cm);
%\draw[blue, thick, fill=white] (h22) circle (0.1cm);
%\draw[blue, thick, fill=white] (f13) circle (0.1cm);
%\draw[blue, thick, fill=white] (f23) circle (0.1cm);
%\draw[blue, thick, fill=white] (h13) circle (0.1cm);
%\draw[blue, thick, fill=white] (h23) circle (0.1cm);
%\draw[blue, thick, fill=white] (f) circle (0.1cm) node[black, thick, above] {$h$};
%\draw[blue, thick, fill=white] (h) circle (0.1cm) node[black, thick, below] {$f$};
\draw[blue, thick, fill=white] (a) circle (0.1cm) node[black, thick, above] {$f$};
\draw[black, thick, rounded corners=5pt] ($(r11)+(-0.5,2)$) rectangle ($(r23)+(0.5,-0.75)$);
\draw [black, dot diameter=4pt, dot spacing=8pt, dots] ($(x2) + (3, 1.25)$) -- ($(x3) + (-.2, 1.25)$);
\node[black, below, ultra thick] at (0.625,1.5) {$A$};
\node[black, below, ultra thick] at ($(r21) + (-1.25,0)$) {$I_1$};
\node[black, below, ultra thick] at ($(r22) + (-1.25,0)$) {$I_2$};
\node[black, below, ultra thick] at ($(r23) + (-1.25,0)$) {$I_n$};

\end{tikzpicture}
\caption{Illustration of a vertex block $H_i$ for some $i \in [k]$. 
%The blue star denotes the induced graph constructed by $C(H_i)$ for the block $H_i$. 
An edge block $H_{i,j}$ for some $1 \leq i < j \leq k$ will have $|E_{i,j}|$-many internal gadgets. % and each of them are indexed by respecting edges.
}
\label{fig:block}
\end{figure}

\noindent Next, we introduce the connector vertices to connect the edge-partition blocks and vertex-partition blocks. 
Let $R = \{r_{i,j}^i, s_{i,j}^i, r_{i,j}^j, s_{i,j}^j \mid 1 \leq i < j \leq k\}$ be the connector vertices. 
The blocks are connected based on the following exclusive and exhaustive cases, and is illustrated in Figure~\ref{fig:connector}:\\
%For each $1 \leq i < j \leq k$, the pair $(r_{i,j}^i, s_{i,j}^i)$ is used to connect the blocks $H_i$ and $H_{i,j}$, and the pair $(r_{i,j}^j, s_{i,j}^j)$ is used to connect the blocks $H_{j}$ and $H_{i,j}$. 
%This is illustrated in . 
%Each connector vertex is made adjacent to a structure of type $T$. 
%{\bf To avoid ambiguity, when $j < i$, the indices $\{i,j\}$ and $\{j,i\}$ are same.} 
For each $i \in [k]$, each $i < j \leq k$ and each $\ell \in [n]$, the edges are described below. 
\begin{itemize}
 \item for each $1 \leq t \leq \ell$, the vertex $a_t$ in the gadget $I_\ell$ of $H_i$ is  adjacent to the vertex $s_{i,j}^i$
 %for all $1 \leq j < i$, and
 \item for each $\ell \leq t \leq n$, the vertex $a_t$ in the gadget $I_\ell$ of $H_i$ is  adjacent to the vertex $r_{i,j}^i$
 %for all $1 \leq j < i$, and
 \end{itemize}
For each $i \in [k]$, each $1 \leq j < i$ and each $\ell \in [n]$, 
\begin{itemize}
\item for each $1 \leq t \leq \ell$, the vertex $a_t$ in the gadget $I_\ell$ of $H_i$ is  adjacent to the vertex $s_{j,i}^i$
%for all $i < j \leq k$, and
\item for each $\ell \leq t \leq n$, the vertex $a_t$ in the gadget $I_\ell$ of $H_i$ is  adjacent to the vertex $r_{j,i}^i$
%for all $i < j \leq k$
\end{itemize}
%\begin{itemize}
%\item for each $1 \leq t \leq \ell$, the vertex $a_t$ in the gadget $I_\ell$ of $H_i$ is made adjacent to the vertex $s_{i,j}^i$ for all $j \in [n]$ with $j \not= i$, and
%\item for each $\ell \leq t \leq n$, the vertex $a_t$ in the gadget $I_\ell$ of $H_i$ is made adjacent to the vertex $r_{i,j}^i$ for all $j \in [n]$ with $j \not= i$.
%\end{itemize}
\noindent
For each $1 \leq i < j \leq k$, and for each $e = u_{i,x}u_{j,y} \in E_{i,j}$, 
\begin{itemize}
\item for each $1 \leq t \leq x$, the vertex $a_t$ in the gadget $I_e$ of $H_{i,j}$ is  adjacent to the vertex $r_{i,j}^i$
\item for each $x \leq t \leq n$, the vertex $a_t$ in the gadget $I_e$ of $H_{i,j}$ is  adjacent to the vertex $s_{i,j}^i$
\item for each $1 \leq t \leq y$, the vertex $a_t$ in the gadget $I_e$ of $H_{i,j}$ is  adjacent to the vertex $r_{i,j}^j$
\item for each $y \leq t \leq n$, the vertex $a_t$ in the gadget $I_e$ of $H_{i,j}$ is  adjacent to the vertex $s_{i,j}^j$
\end{itemize}
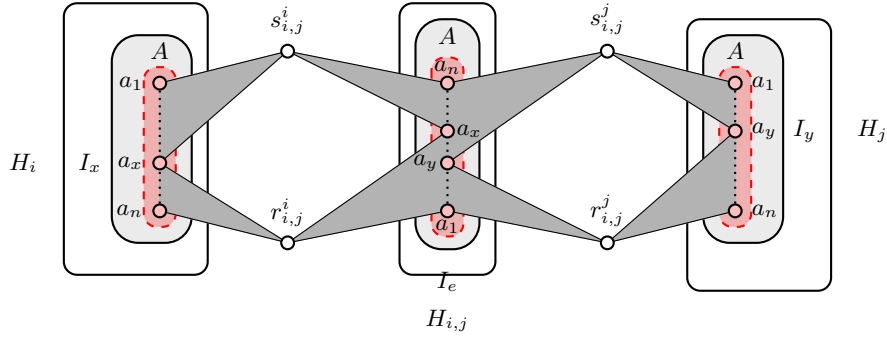
\begin{figure}[h]
\centering
\begin{tikzpicture}[scale=0.85]

\coordinate (a1) at (2,3.5);
\coordinate (a2) at (2,2.25);
\coordinate (a3) at (2,1.5);
\coordinate (b1) at (11,3.5);
\coordinate (b2) at (11,2.75);
\coordinate (b3) at (11,1.5);
\coordinate (c1) at (6.5,3.5);
\coordinate (cx) at (6.5,2.75);
\coordinate (cy) at (6.5,2.25);
\coordinate (c3) at (6.5,1.5);
\coordinate (si) at (4,4);
\coordinate (ri) at (4,1);
\coordinate (sj) at (9,4);
\coordinate (rj) at (9,1);

%Hi
\draw[thick, black, fill=white, rounded corners=5pt ] ($(a1) + (-1.5, 1.25)$) rectangle ($(a3)+(0.75,-1)$);
\draw[black, thick, fill=lightgray, fill opacity=0.3, rounded corners=10pt] ($(a1)+(-0.75,+0.75)$) rectangle ($(a3)+(0.5,-0.5)$);
\draw[red, fill=red, fill opacity=0.25, dashed, thick,rounded corners=5pt ] ($(a1)+(-0.25,+0.25)$) rectangle ($(a3)+(0.25,-0.25)$);
\filldraw[draw=gray!60, fill=gray!60] (a1) -- (si) -- (a2) -- (a1) -- cycle;
\filldraw[draw=gray!60, fill=gray!60] (a3) -- (ri) -- (a2) -- (a3) -- cycle;
\draw[black] (a1) -- (si);
\draw[black] (a2) -- (si);
\draw[black] (a2) -- (ri);
\draw[black] (a3) -- (ri);
\draw[black, thick, dotted] (a1) -- (a3);
\draw[thick, black, fill=red!25] (a1) circle (0.1cm) node[left=0.1cm] {$a_1$} node[above=0.2cm] {$A$};
\draw[thick, black, fill=red!25] (a2) circle (0.1cm) node[left=0.1cm] {$a_{x}$} node[left=0.65cm] {$I_x$} node[left=1.5cm] {$H_i$};
\draw[thick, black, fill=red!25] (a3) circle (0.1cm) node[left=0.1cm] {$a_n$};

%Hj
\draw[thick, black, fill=white, rounded corners=5pt ] ($(b1) + (-0.75, 1)$) rectangle ($(b3)+(1.5,-1.25)$);
\draw[black, thick, fill=lightgray, fill opacity=0.3, rounded corners=10pt] ($(b1)+(-0.5,+0.75)$) rectangle ($(b3)+(0.75,-0.5)$);
\draw[red, fill=red, fill opacity=0.25, dashed, thick,rounded corners=5pt ] ($(b1)+(-0.25,+0.25)$) rectangle ($(b3)+(0.25,-0.25)$);
\filldraw[draw=gray!60, fill=gray!60] (b1) -- (sj) -- (b2) -- (b1) -- cycle;
\filldraw[draw=gray!60, fill=gray!60] (b3) -- (rj) -- (b2) -- (b3) -- cycle;
\draw[black] (b1) -- (sj);
\draw[black] (b2) -- (sj);
\draw[black] (b2) -- (rj);
\draw[black] (b3) -- (rj);
\draw[black,thick, dotted] (b1) -- (b3);
\draw[thick, black, fill=red!25] (b1) circle (0.1cm) node[right=0.1cm] {$a_1$} node[above=0.2cm] {$A$};
\draw[thick, black, fill=red!25] (b2) circle (0.1cm) node[right=0.1cm] {$a_{y}$} node[right=0.65cm] {$I_y$} node[right=1.5cm] {$H_{j}$};
\draw[thick, black, fill=red!25] (b3) circle (0.1cm) node[right=0.1cm] {$a_n$};

%Hij
\draw[thick, black, fill=white, rounded corners=5pt ] ($(c1) + (-0.75, 1.25)$) rectangle ($(c3)+(0.75,-1)$);
\draw[black, thick, fill=lightgray, fill opacity=0.3, rounded corners=10pt] ($(c1)+(-0.5,1)$) rectangle ($(c3)+(0.5,-0.6)$);
\draw[red, fill=red, fill opacity=0.25, dashed, thick,rounded corners=5pt ] ($(c1)+(-0.25,+0.4)$) rectangle ($(c3)+(0.25,-0.4)$);
\filldraw[draw=gray!60, fill=gray!60] (c1) -- (si) -- (cx) -- (c1) -- cycle;
\filldraw[draw=gray!60, fill=gray!60] (c1) -- (sj) -- (cy) -- (c1) -- cycle;
\filldraw[draw=gray!60, fill=gray!60] (c3) -- (ri) -- (cx) -- (c3) -- cycle;
\filldraw[draw=gray!60, fill=gray!60] (c3) -- (rj) -- (cy) -- (c3) -- cycle;
\draw[black] (c1) -- (si);
\draw[black] (c1) -- (sj);
\draw[black] (cx) -- (ri);
\draw[black] (cx) -- (si);
\draw[black] (cy) -- (rj);
\draw[black] (cy) -- (sj);
\draw[black] (c3) -- (ri);
\draw[black] (c3) -- (rj);
\draw[thick, black, dotted] (c1) -- (c3);
\draw[thick, black, fill=red!25] (c1) circle (0.1cm) node[above] {$a_n$} node[above=0.4cm] {$A$};
\draw[thick, black, fill=red!25] (cx) circle (0.1cm) node[right] {$a_{x}$} node[below=1.75cm] {$I_e$} node[below=2.25cm] {$H_{i,j}$};
\draw[thick, black, fill=red!25] (cy) circle (0.1cm) node[left] {$a_{y}$};% node[right=0.65cm] {$I_y$} node[below=1.5cm] {$H_{j}$};
\draw[thick, black, fill=red!25] (c3) circle (0.1cm) node[below] {$a_1$};

\draw[thick, black, fill=white] (si) circle (0.1cm) node[above=0.1cm] {$s^i_{i,j}$};
\draw[thick, black, fill=white] (ri) circle (0.1cm) node[above=0.1cm] {$r^i_{i,j}$};
\draw[thick, black, fill=white] (sj) circle (0.1cm) node[above=0.1cm] {$s^j_{i,j}$};
\draw[thick, black, fill=white] (rj) circle (0.1cm) node[above=0.1cm] {$r^j_{i,j}$};
\end{tikzpicture}
\caption{
An illustration of the connector vertices $s_{i,j}^i$, $r_{i,j}^i$, $s_{i,j}^j$ and $r_{i,j}^j$ connect the blocks $H_i$ and $H_{i,j}$, and, $H_{j}$ and $H_{i,j}$ for some $1 \leq i < j \leq k$. The edge $e$ represented in the gadget $I_e$ is $u_{i,x}u_{j,y} \in E_{i,j}$. }
\label{fig:connector}
\end{figure}
This completes construction of the graph $H$ with ${\cal O}(mn^2)$ vertices and ${\cal O}(mn^3)$ edges. We next bound the pathwidth of the graph $H$ as a polynomial function of $k$. 
%\begin{definition}[Compatible gadget pairs]
%For $1 \leq 1 < j \leq k$, given an edge $e \in E_{i,j}$ and $u_{i,x} \in V_i$, the gadgets $I_e$ in $H_{i,j}$ and $i_i$ in $H_i$ are said to be compatible if $e$ hits $u_{i,x}$ in $G$.
%\end{definition}
%\begin{claim}
%Let $S \subseteq V(H)$ be a set. 
%Given $1 \leq i < j \leq k$, let $S \cap V(H_i)$ contain 
%\end{claim}
\begin{claimnew}
\label{claim:BlockPathwidth}
The pathwidth of a block $B$ is at most five. 
\end{claimnew}
%\begin{comment}
\begin{claimproof}
If we remove the vertex $f$ from the block $B$, then the resulting graph is a disjoint collection of gadgets and a tree of height two. 
We know that the pathwidth of a gadget is four from Claim~\ref{claim:gagdetPathwidth}, and the pathwidth of a tree of height two is two. 
Let $(\T',\X')$ be a path decomposition of $B-\{f\}$ with pathwidth four. 
Thus adding $f$ into all bags of $(\T', \X')$ gives a path decomposition for the block $B$, and thus the pathwidth of the block is at most five. 
\end{claimproof}
%\end{comment}

\begin{lemma}%$\diamondsuit$
\label{lem:HPathwidth}
The pathwidth of the graph $H$ is at most $4{k \choose 2} + 5$. 
\end{lemma}
%\begin{comment}
\begin{proof}%[{\bf Proof of Lemma~\ref{lem:HPathwidth}}]
Removal of the connector vertices from $H$ results in a collection of disjoint blocks. 
From Claim~\ref{claim:BlockPathwidth}, the pathwidth of a block is five. 
Let $(\T',\X')$ be a path decomposition of $H-R$ with pathwidth five. 
%Since the pathwidth of a block is five from Claim~\ref{claim:BlockPathwidth}. 
Therefore, adding all connector vertices to the path decomposition $(\T', \X')$ gives a path decomposition for the graph $H$ with pathwidth at most $4{k \choose 2}+5$.  
\end{proof}
%\end{comment}

\subsubsection{Properties of a feasible solution for the MMDS instance $(H, k')$.}
Let $S$ be a feasible solution for the MMDS instance  $(H, k')$. We state the following properties of the set $S$. In all the arguments below, we crucially use the property that for each $u \in V(H)$, $M(u,S) \leq n+1$.
%to being a feasible solution to the instance $(H, k')$ of the MMDS problem. 

\begin{claimnew}%$\diamondsuit$
\label{claim:structT}
%Let $S$ be a feasible solution for the instance ($H, k'$) of the MMDS problem. 
For each block $B$ in the graph $H$, $C'(B) \subseteq S$. 
\end{claimnew}
\begin{claimproof}
By construction of graph $H$, for each $1 \leq p \leq n+1$, the vertex $\ch_p$ must be in the set $S$ since it has $n+2$ vertices of degree one as neighbors.  Otherwise, its membership will be at least $n+2$, contradicting that $S$ is a feasible solution for $(H,k')$. Hence the claim.
\end{claimproof}

\begin{claimnew}%$\diamondsuit$
\label{claim:excludea}
%For each block in the graph $H$, the vertex $a$ in the block is not in the solution. 
For each block $B$ in $H$, the vertices $f$ and $f'$ in $B$ are not in the set $S$. 
\end{claimnew}

\begin{claimproof}
We know that $f$ is made adjacent to $f'$, and $f'$  is adjacent to each vertex in $C'(B)$. 
From  Claim~\ref{claim:structT}, we know that $C'(B)$ is a subset of $S$.
Thus, $n+1$ neighbors of $f'$ is in $S$. 
If either $f$ or $f'$ is in the set $S$, then $M(f',S)$ is $n+2$. This contradicts the feasibility of the set $S$. Hence the claim.
\end{claimproof}

\begin{claimnew}%$\diamondsuit$
\label{claim:AZeroFull}
For each gadget of type $I$ in each block $B$ in the graph $H$, either $A \cap S = A$ or $A \cap S = \emptyset$. 
\end{claimnew}
\begin{claimproof}
We prove this by contradiction. 
Assume that $\emptyset \subsetneq A \cap S \subsetneq A$. 
Let $J = \{j \in [n] \mid a_j \in S\}$, that is $J$ is the index of the elements in  $A \cap S$. Note that by our premise $J$ is non-empty and it is not all of $[n]$.  Since $J$ is a strict subset of $[n]$, we observe that
the vertex $h_1$ is in $S$. This is because,   for each $i \in [n] \setminus J$, $h_1$ and $a_i$ is connected by the gadget of type $\D$.  If both $a_i$ and $h_1$ are not in $S$, then  the $n+2$ neighbours in the gadget of type $\D$ containing the edge $\{a_i, h_1\}$ will be in $S$, and thus $M(a_i,S)$ and $M(h_1,S)$ are both at least $n+2$. This  violates the hypothesis that for each $u \in V(H)$, $M(u,S) \leq n+1$.  
We now consider two cases, one in which the vertex $h_2$ is in $S$ and the other in which $h_2$ is not in $S$. \\
First, we consider $h_2 \in S$. 
For each $i \in [n]$, by using the same argument which we used for $a_i$ and $h_1$, it follows that at least one of the $a_i$ or $d_i$ is in the set $S$ since $a_i$ and $d_i$ are both in a  gadget of type $\D$.
Therefore, for each $i \in [n]\setminus J$, the vertex $d_i$ is in $S$. 
That is $|D \cap S| \geq n-|J|$. 
Consequently, using the fact that $h_1 \in S$ and the premise that $h_2 \in S$, it follows that the membership of $h_1$ is 
\[M(h_1,S) \geq |A \cap S| + |D\cap S| + 2 \geq |J| + n-|J| + 2 \geq n+2.\]
This contradicts the feasibility of $S$. \\
Next we consider the case that $h_2$ is not in $S$. 
For each $i \in [n]$, $d_i$ is in $S$ since $d_i$ and $h_2$ are in a gadget of type $\D$.
Then, the $N[h_1] = (A \cap S) \cup D \cup \{h_1\}$. Further, we know that $J$ is a non-empty set, and thus, the membership of $h_1$ is 
\[M(h_1,S) \geq |A \cap S| + |D| + 1 \geq |J| + n + 1 \geq n+2.\]\ 
Therefore, our assumption that that $A \subsetneq S$ and $A \cap S \neq \emptyset$ is wrong.  Therefore, either the set $A$ is completely included in the set $S$ or completely excluded from the set $S$. 
\end{claimproof}

\begin{claimnew}%$\diamondsuit$
\label{claim:unqiueGadget}
For each block $B$ in the graph $H$, there exists a unique gadget of type $I$ in the block $B$ such that the set $A$ in the gadget is in $S$. 
\end{claimnew}

\begin{claimproof}
The vertices $f$ and $f'$ in $B$ are not in the solution $S$ due to Claim~\ref{claim:excludea}. 
The Claim~\ref{claim:AZeroFull} states that either the set $A$ in any gadget is completely included in the set $S$ or completely excluded in the set $S$. 
If for each gadget in $B$, the set $A$ is not in $S$ then the vertex $f$ is not dominated by $S$. 
This contradicts the feasibility of $S$. 
If the set $A$ of at least two gadgets in the block $B$ are in the set $S$, then the membership of $f$ will be $2n > n+1$. 
This contradicts the feasibility of $S$. 
Thus, there exists an unique gadget $I$ in each block such that the set $A$ in $I$ is in $S$. 
\end{claimproof}

\noindent Using these properties in the following two lemmas, we prove the correctness of the reduction.
%\begin{claim}
%In each gadget $I$ in the graph $H$, either $B\cup D \cup \{h_1,f_2\} \cup \ch(p') \cup \ch(q')$ is
%$A\cup C \cup \{h_2,f_1\}\cup \ch(p') \cup \ch(q')$
%\end{claim}
%\subsubsection{Equivalence between the instance input to the reduction and the instance output by the reduction}
%Now we show the equivalence of both the problems. 
%More precisely, the graph $G$ has a $k$-clique if and only if $H$ has a dominating set with membership at most ^%$k'$. 
\begin{lemma}%$\diamondsuit$
\label{lem:mcc:mmds}
If $(G,k)$ is a YES-instance of the \mcc problem, then $(H,k')$ is a YES-instance of the MMDS problem. %\mmds problem.
\end{lemma}

\begin{proof}
Let $K = \{u_{i,x_i} \mid i \in [k]\}$ be a $k$-clique in $G$. 
That is, for each $i \in [k]$, $x_i$-th vertex of the partition $V_i$ is in the clique. 
Now we construct a feasible solution $S$ for the instance $(H,k')$ of the MMDS problem. 
The set $S$ consists of the following vertices.
For each $i \in [k]$,
\begin{itemize}
\item for each $\ell \in [n]$ with $\ell \not= x_i$, add $D \cup \{h_1\}$ % \cup \bigcup_{j \in [n+1]}\{ch_{p'}^j\} \cup \bigcup_{j \in [n+1]}\{ch_{q'}^j\}$ 
in the gadget $I_\ell$ in the vertex-partition block $H_i$ to $S$, and
\item in the gadget $I_{x_i}$ in the vertex-partition block $H_i$, add $A \cup \{h_2\}$ to $S$, and 
\item add $C'(H_i)$ to $S$. 
\end{itemize}
For each $1 \leq i < j \leq k$, 
\begin{itemize}
\item for each edge $e \in E_{i,j}$ with $e \not= u_{i,x_i}u_{j,x_j}$, add $D \cup \{h_1\}$ in the gadget $I_e$ the edge-partition block $H_{i,j}$ to $S$, and 
\item for the edge $e=u_{i,x_i}u_{j,x_j}$, add $A \cup \{h_2\}$ in the gadget $I_e$ the edge-partition block $H_{i,j}$ to $S$, and 
\item add $C'(H_{i,j})$ to $S$. 
\end{itemize}
We show that $S$ is a feasible solution to the MMDS problem in $H$ for membership value $k'=n+1$. 

\noindent
First we show that the set $S$ is a dominating set in $H$. 
In each gadget in each block, we have added either $D \cup \{h_1\}$ or $A \cup \{h_2\}$ into $S$. 
Therefore, in every gadget of type $\D$ at least one head is in $S$. 
That is, $S$ dominates every vertex which is part of some gadget of type $\D$. 
Since every vertex in a gadget of type $I$ is part of some gadget of type $\D$, the gadget of type $I$ is dominated by $S$. 
Thus, every gadget of type $I$ is dominated by $S$. 

\noindent
Then we consider the vertices outside any gadget of type $I$. In any block $B$, this is the set $C(B)$.  
In each block $B$ in $H$, % (both vertex and edge blocks), 
from Claim~\ref{claim:unqiueGadget}, 
vertices in the set $A$ of exactly one gadget of type $I$ is in $S$. Each of these vertices dominate $f$.
%Then they are dominating the vertex $f$ in the block outside gadgets.  
All other vertices in the block which are outside the gadgets are dominated by $C'(B)$ which is a subset of $S$ by definition. 
%We left with connector vertices. 
For $1 \leq i < j \leq k$, consider a connector vertex pair $(s_{i,j}^i, r_{i,j}^i)$ that connects the blocks $H_i$ and $H_{i,j}$. 
Both vertices were made adjacent to the vertices in the set $A$ of each gadget in the block $H_i$. 
Since $S$ contains the set $A$ in the gadget $I_{x_i}$ of $H_i$, both connector vertices are dominated. 
Thus all the connector vertices are dominated by $S$.  Therefore, $S$ is a dominating set of $H$. 

Next we show that the membership of any vertex $u \in V(H)$ in $S$ is $k'$, that is, we show that $M(u,S) = n+1$. 
Observe that the vertices in any gadget of type $I$ are solely dominated by the vertices of $S$ which are inside the gadget.
%in the solution from the gadget. 
The maximum membership of $n+1$ is achieved by the vertices $h_1$ and $h_2$ for any gadget $I$. 
Therefore, the membership of any vertex in a gadget of type $I$ is $n+1$. 
In each block $B$, among the vertices $C(B)$, the maximum membership of value $n+1$ is achieved by the vertices $f$ and $f'$. 

\noindent
We next show crucially that the membership of the connector vertices is at most $n+1$. 
For each $1 \leq i < j \leq k$, consider the  edge $e=u_{i,x_i}u_{j,x_j} \in E_{i,j}$.
%hits $u_{i,x_i} \in V_i \cap K$. 
By construction of the set $S$, we picked the set $A$ only from the gadgets $I_e$ in $H_{i,j}$, $I_{x_i}$ from $H_i$, and $I_{x_j}$ from $H_j$.  From the reduction and the definition of $S$, it is clear that for all the other gadgets in a block,  the vertices in $S$ are not adjacent to the connector vertices.
%The gadgets $I_e$ and $I_{x_i}$ are compatible. 
Therefore, the $M(s_{i,j}^i,S)$  is $x_i + (n-x_i+1) = n+1$, and $M(r^i_{i,j},S)$ is $(n-x_i+1)+x_i = n+1$. Next, we consider the membership of $s_{i,j}^j = x_j + (n-x_j+1)$ and the membership of $r_{i,j}^j$ is $x_j + (n-x_j+1) = n+1$.  These membership values can be seen clearly from Figure~\ref{fig:connector}.  
Hence, the membership of any vertex in $V(H)$ is $n+1$. 
Thus, the instance $(H, k')$ is a YES-instance of the MMDS problem. 
\end{proof}

%In order to prove the other direction of the equivalence, we state some properties of a feasible solution to the instance  $(H,k')$ of the MMDS problem. 
%Let $S$ be a feasible solution. % for the instance MMDS($H,k'$). 

\begin{lemma}
\label{lem:mmds:mcc}
If $(H,k')$ is a YES-instance of the MMDS problem, then $(G,k)$ is a YES-instance of the \mcc~problem. %\mmds problem.
\end{lemma}
\begin{proof}
Let $S$ be a feasible solution to the instance ($H,k'$) of the MMDS problem. 
For each $i \in [k]$, let $I_{x_i}$ be the unique gadget for some $x_i \in [n]$, where the set $A$ of $I_{x_i}$ is in $S$. 
For each $1 \leq i < j \leq k$, let $I_{e}$ be the unique gadget for some $e = u_{i,x_i'}u_{j,x_j'} \in E_{i,j}$, where the set $A$ of $I_e$ is in $S$. 
The existence of such gadgets are ensured by Claim~\ref{claim:unqiueGadget}. 
Let $K = \{u_{i,x_i} \mid i \in [k]\}$. 
We show that the set $K$ is a clique in $G$ as follows. 
Observe that we picked one vertex from each partition $V_i$ for $i \in [k]$. 
Next we show that for each $1 \leq i < j \leq k$, there is an edge $u_{i,x_i}u_{j,x_j} \in E(G)$. 
Let $i, j \in [k]$ such that $i < j$. 
The vertex $s_{i,j}^i$ is adjacent to $x_i$ vertices in $I_{x_i}$ from $H_i$, and $n-x_i'+1$ vertices in $I_e$ from $H_{i,j}$. 
The vertex $r_{i,j}^i$ is adjacent to $n-x_i+1$ vertices in $I_{x_i}$ from $H_i$, and $x_i'$ vertices in $I_e$ from $H_{i,j}$. 
Then, the membership of the connector vertices $r^i_{i,j}$ and $s^i_{i,j}$ in $S$ are
\[M(r^i_{i,j}, S) \geq (n-x_i+1) + (x_i')\geq n+x_i'-x_i+1\textrm{ , and}\] 
\[M(s^i_{i,j}, S) \geq x_i + (n-x_i'+1) \geq n+x_i-x_i'+1.\]
%These membership value can be seen clearly from 
Further, the membership of the vertices is at least one and at most $n+1$, that is $1 \leq M(r_{i,j}^i, S), M(s_{i,j}^i, S) \leq n+1$. 
Therefore, $n+1 \geq n+x_i'-x_i + 1\Longrightarrow x_i \geq x_i'$ and $n+1 \geq n+x_i-x_i' + 1\Longrightarrow x_i' \geq x_i$. 
Thus, we have $x_i = x_i'$. 
Similarly, we will get $x_j = x_j'$. 
Therefore, by construction of the graph $H$, there is an edge $u_{i,x_i}u_{j,x_j} \in E(G)$. 
Thus, the set $K$ is a feasible solution for the instance ($G,k$) of the \mcc problem. 
\end{proof}

Thus, we conclude the section with the proof of Theorem~\ref{thm:twhard}. 
\begin{proof}[Proof of Theorem~\ref{thm:twhard} ]
On an instance $(G,k)$ of \mcc the reduction constructs $(H,k'=n+1)$ in polynomial time. 
%Note that $n$ is the size of each partition in $G$.   
From Lemma~\ref{lem:HPathwidth} we know that the pathwidth of $H$ is a quadratic function of $k$.  Finally, from Lemma~\ref{lem:mcc:mmds} and Lemma~\ref{lem:mmds:mcc} it follows that the MMDS instance  $(H,k')$ output by the reduction is equivalent to the \mcc instance $(G,k)$ that was input to the reduction. Since \mcc is known to be W[1]-hard for the parameter $k$, it it follows that the MMDS problem is W[1]-hard with respect to the parameter pathwidth of the input graph. 
%Therefore, under exponential time hypothesis the MMDS problem cannot be solved in $n^{o(\pw(G))}$ time. 
\end{proof}

\section{W[1]-hardness in split graphs}
\label{sec:splithard}
In this section we prove that \mmds is \textsc{W}[1]-hard on split graphs when parameterized by the membership parameter $k$. We prove this result by demonstrating a parameterized reduction from \mcis (MIS) to \textsc{Minimum Membership Dominating Set}. \mcis requires finding a colorful independent set of size $k$ and is known to be \textsc{W}[1]-hard for the parameter solution size~\cite{FELLOWS200953} %(See Appendix for formal problem definition).
\begin{tcolorbox}
\mcis\\
\textsf{Input}: A positive integer $k$, and a $k$-colored graph $G$.\\
{\sf Parameter}: $k$\\
{\sf Question}: Does there exist an independent set of size $k$ with one vertex from each color class?
\end{tcolorbox}

Let $(G=(V,E),k)$ be an instance of the \mcis problem. 
Let $V=(V_1,\ldots, V_k)$ be the partition of the vertex set $V$, where vertices in set $V_i$ belong to the $i^{th}$ color class, $i\in[k]$. We now show how to construct a split graph $H=(V'\cup V'', E')$ such that if $(G,k)$ is a YES instance, then $H$ has a dominating set with maximum membership $k$. $V'$ refers to the clique partition of $H$ and $V''$ consists of the partition containing a set of independent vertices.

\noindent
\textbf{Construction of graph $H=(V'\cup V'', E')$:}

\begin{figure}
\centering
\begin{tikzpicture}[scale=0.8, every node/.style={scale=0.8}]
%\draw(3,2)--(8,2);
%\draw(8,0.5)--(12,0.5);
%\draw (8.7,1.5)--(9.7,1.5);
%\draw(10.2,1.5)--(11.2,1.5);
%\draw (12,1.7)--(14,1.7);
%\fill[orange] (10.3,3.2) -- (3,6.7)--(3,5)--(10.3,3.2);
\filldraw[color=red!60, fill=orange!25,  thick](10.3,3.2) -- (3,6.7)[dotted]--(3,5)--(10.3,3.2);
%\fill[cyan] (10.3,3.2) -- (5,6.7)--(5,5)--(10.3,3.2);
\filldraw[color=blue!60, fill=cyan!25,  thick](10.3,3.2) -- (5,6.7)--(5,5)--(10.3,3.2);

\draw (3,6) ellipse (0.5cm and 1.75cm);  %V_i's
\draw (5,6) ellipse (0.5cm and 1.75cm);
\draw[very thick,dotted] (6,6)--(7,6);
\draw (8,6) ellipse (0.5cm and 1.75cm);
%\draw[dotted] (9,1)--(11,1);
\draw (3.5,0.5) rectangle (2.5,3.5);	 %U_i's	
\draw (5.5,0.5) rectangle (4.5,3.5);
\draw[very thick,dotted] (6,2)--(7,2);
\draw (8.5,0.5) rectangle (7.5,3.5);
\fill[black](3,7) circle (0.1);         % u
\draw(3,7.3) node{$u$};
\fill[black](3,5) circle (0.1);         
 \fill[black](3,6.7) circle (0.1);    %node below u
\fill[black](5,7) circle (0.1);       % v 
\draw(5,7.3) node{$v$};
 \fill[black](5,6.7) circle (0.1);	  %node below v
\draw (3,7)--(5,7);
 \fill[black](5,5) circle (0.1);

%\fill[black](5,5) circle (0.1);
%\fill[black](8,7) circle (0.1);
%\fill[black](8,5) circle (0.1);
\fill[black](3,3) circle (0.1);
\draw[very thick,dotted] (3, 2.5)--(3,1.5);
\fill[black](3,1) circle (0.1);
\draw (3,7) .. controls (1.5,5) .. (3,3);
\draw (3,7) .. controls (0.5,4) .. (3,1);
\draw (3,5) .. controls (2.7,4) .. (3,3);
\draw (3,5) .. controls (1.5,3) .. (3,1);

%\draw (10,6.5) rectangle (13,5.5);     %W_pq
%\draw [decorate,decoration={brace,amplitude=5pt,mirror,raise=-4ex}]
%  (10,6.9) -- (13,6.9) node[midway,xshift=-3em]{$k\choose 2$};
\fill[black](11.3,6) circle (0.1);
\draw(11.3,6.3) node{$w$ };

%\fill[black](12.7,6) circle (0.1);
\draw (11.3,6) -- (10.3, 3.2);
\draw (11.3,6) -- (12.3, 3.2);
%\draw(11,7) node{$k \choose 2$ };
\draw (11.3,6) -- (10.3, 0.8);
\draw (11.3,6) -- (12.3, 0.8);
\fill[black](10.3,0.8) circle (0.1);
\fill[black](12.3,0.8) circle (0.1);
\draw [dotted](10.4,0.8) -- (12.2, 0.8);
\draw[very thick,dotted](2,4) rectangle (14,9);
\draw(3,8) node{$V_1$};
\draw(5,8) node{$V_2$};
\draw(8,8) node{$V_k$};

\draw(3,0) node{$U_1$};
\draw(5,0) node{$U_2$};
\draw(8,0) node{$U_k$};
\draw [decorate,decoration={brace,amplitude=8pt},xshift=-4pt,yshift=0pt]
(1,0.5) -- (1,3.5) node [black,midway,xshift=-0.8cm] 
{ $k+1$};

\draw (11.3,3.2) ellipse (1.35cm and 0.4cm);         %D_pq
\draw(10.6,3) node{$x_{uv}$};
\fill[black](10.3,3.2) circle (0.1);
\fill[black](12.3,3.2) circle (0.1);
\draw [dotted] (10.5,3.2) -- (12,3.2);
\draw[very thick,dotted] (11.3, 1.5)--(11.3, 2.5);
\draw (11.3,0.8) ellipse (1.30cm and 0.4cm);
%\draw (10,0.5) rectangle (13,1.5);
\draw(13.2,3) node{$D_{12}$};
\draw(13.2,1) node{$D_{pq}$};
 
\draw [decorate,decoration={brace,amplitude=8pt,mirror,raise=4pt},yshift=0pt]
(14,0.5) -- (14,3.5) node [black,midway,xshift=0.8cm]
{ ${k \choose 2}$};
%\draw(10,0.3) node{$u'$};
%\draw(9.2,1.7) node{$v_1$};
%\draw(10.7,1.7) node{$v_2$};

%\draw(5,-1) node{$z_m$};
%\draw(7,-1) node{$z_t$};
%\draw(8,-0.75) node{$r'$};
%\draw(9,-0.75) node{$r'+1$};
\end{tikzpicture}
\caption{Construction of graph $H$}
\label{fig:splithardness}
\end{figure}
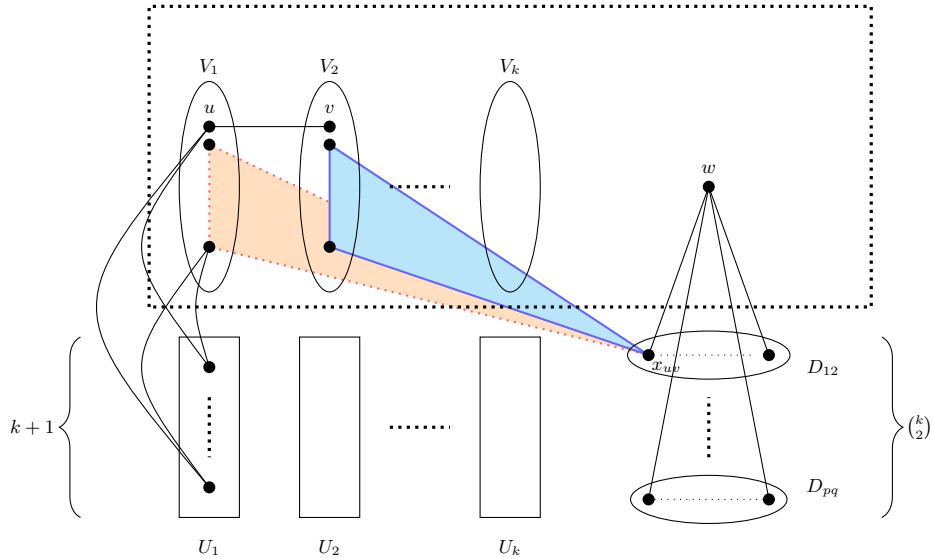
For each vertex in $V$ we introduce a vertex in the clique $V'$ as in the input instance. Additionally we add a vertex  $w$ to $V'$.  Edges are added among each pair of vertices in $V'$.  
%Each vertex $w_{pq}$ represents the set of edges between the vertices in sets $V_p$ and $V_q$ in $G$.
The  set $V''$ in $H$ is an independent set, and it consists of a set of vertices denoted by $U$, a set of vertex sets denoted by $\mathcal{D}=\{D_{pq}\mid p,q\in [k], p < q\}$. 
The vertex set $U$ comprises a partition of $k$ vertex sets, $U=\{U_i\mid  i\in [k]\}$, and $|U_i| = k+1$. For each edge between a vertex $u\in V_p$ and $v\in V_q$ in $G$, we introduce a vertex $x_{uv}$ in the set $D_{pq}$. Thus, the vertex set of $H$, $V(H)=V'\cup V''$, where $V'$ induces a clique and $V''$ induces an independent set.\\
\noindent %The set of edges $E'$ contains $E$ and the following additional edges. 
%Since vertex sets $V$ and $W$ are part of the clique in $H$, there exists an edge between each pair of vertices in $V\cup W$. Let $E_1$ be the set of edges with both endpoints in $V'$, $E_1 = \{(u,v) ~|~ u,v\in V', u\neq v\}$. There are no edges within each set $U_i$, $i\in [k]$.
The remaining edges, other than those in clique $V'$, are described as follows: $V_i\uplus U_i$ forms a complete bipartite graph, %i.e. for each value of $i\in[k]$ there exists an edge between each vertex $v\in V_i$ and $u\in U_i$. Lets $E_2$ be the set of edges between vertex sets $V_i$ and $U_i$, $E_2 = \{(u,v) ~|~ u \in V_i, v\in U_i, 1\leq i \leq k\} $. 
 vertex $w$ is made adjacent to all vertices in the set $D$, 
%Let $E_3$ be the set of edges between vertices $w_{pq}$ and $D_{pq}$, $p,q\in [k], p < q$, $E_3 = \{(w_{pq},x_{uv})~\mid ~ x_{uv}\in D_{pq}, u\in V_p, v\in V_q, 1\leq p,q\leq k, p < q\}$.  
each vertex $x_{uv}\in D_{pq}$ is made adjacent to every vertex in $V_p\setminus\{u\}$ and $V_q\setminus\{v\}$.
%Let $E_4$ be these set of edges, $E_4 = \big\{(y,x_{uv})~\mid ~y\in V_p, x_{uv}\in D_{pq}, u\in V_p, v\in V_q, y\neq u, 1\leq p,q\leq k, p\neq q\big\}\cup\big\{(z,x_{uv})~\mid ~z\in V_q, x_{uv}\in D_{pq}, u\in V_p, v\in V_q, z\neq v, 1\leq p,q\leq k, p < q\big\}$. Hence, the set of edges in graph $H$, $E(H)=E'=E_1\cup E_2\cup E_3\cup E_4$.
%Let $U = (U_1,\ldots, U_k)$ where each $U_i$ is an independent set of $k+1$ vertices. Let $D_{pq}=\{x_{uv}~\mid ~ u\in V_p, v\in V_q, (u,v)\in E\}$.
%\begin{itemize}
%\item $V' = V\cup\{w_{pq}~|~1\leq p,q\leq k, p < q\}$. 
%\item $V'' = U\cup \{D_{pq}~|~1\leq p,q\leq k, p < q\}$
%\item $E_1 = \{(u,v) ~|~ u,v\in V', u\neq v\}$
%\item $E_2 = \{(u,v) ~|~ u \in V_i, v\in U_i, 1\leq i \leq k\} $
%\item $E_3 = \{(w_{pq},x_{uv})~\mid ~ x_{uv}\in D_{pq}, u\in V_p, v\in V_q, 1\leq p,q\leq k, p < q\}$
%\item $E_4 = \big\{(y,x_{uv})~\mid ~y\in V_p, x_{uv}\in D_{pq}, u\in V_p, v\in V_q, y\neq u, 1\leq p,q\leq k, p\neq q\big\}\cup\big\{(z,x_{uv})~\mid ~z\in V_q, x_{uv}\in D_{pq}, u\in V_p, v\in V_q, z\neq v, 1\leq p,q\leq k, p < q\big\}$
%\item $E' = E_1\cup E_2\cup E_3 \cup E_4$
%\end{itemize}
The above construction is depicted in Figure~\ref{fig:splithardness}.
Next we show the correctness of the reduction from the instance $(G,k)$ of MIS to the instance $(H,k)$ of MMDS.
\begin{lemma}%$\diamondsuit$
\label{lem:mcis-mmds}
If $(G,k)$ is a YES instance of the \mcis problem then $(H,k)$ is a YES instance of the \mmds problem.
\end{lemma}
\begin{proof}
Let $(G,k)$ be a YES instance of \mcis, and $S = \{v_1,\ldots, v_k\}$ be a solution to  $(G,k)$ where $v_i\in V_i,~i\in [k]$. We show that the vertices in $V(H)$ that correspond to the set $S$, denoted by $S'$, form a dominating set in $H$ with membership value $k$. We start by showing that $S'$ is a dominating set for $H$. Observe that since $S'\subset V'$ and $V'$ induces a clique, $S'$ dominates all vertices in the set $V'$. For each $i\in[k]$, the vertex $v_i\in S'$ dominates all vertices in $U_i$, since $V_i$ and $U_i$ together form a complete bipartite graph. Consider a vertex $x_{uv}\in D_{pq}$ where  $ u\in V_p, v\in V_q$, which represents the edge $(u,v)\in G$. The vertex $x_{uv}$ is connected to every vertex in $V_p\setminus \{u\}$ and $V_q\setminus \{v\}$. Since $S$ is an independent set, both $u$ and $v$ cannot belong to $S$. Without loss of generality, let $u\in S$, then $\exists v'_{v'\in S'} \in V_q\setminus \{v\}$ which dominates $x_{uv}$. This holds true for all vertices in the set $\{x_{uv}\in D_{pq}\mid1\leq p,q\leq k, p\neq q\}$. Thus $S'$ is a dominating set for $H$.\\

\noindent Secondly, we observe that the membership constraint $k$ is satisfied by the dominating set $S$ due to the fact that $|S| =  k$. It follows that for all vertices $v$ in $H$, $N[v]\cap S \leq k$. Hence it is proved that if $(G,k)$ is a YES instance of the \mcis problem then $(H,k)$ is a YES instance of the \mmds problem.
\end{proof}

\begin{lemma}\label{lem:mmds-mcis}
If $(H,k)$ is a YES instance of the \mmds problem then $(G,k)$ is a YES instance of the \mcis problem.
\end{lemma}
\begin{proof}
Let $(H,k)$ be a YES instance and $S$ be a feasible solution for \mmds in the graph $H$ with membership value $k$. Since $S$ is a $k$ membership dominating set of $H$, $S$ exhibits the following properties.
\begin{enumerate}
\item $|S\cap V_i| = 1 , i\in [1,k]$.\\
 		At least one vertex from each of the sets $V_i$ must belong to $S$. Otherwise in order to dominate $U_i$, all $k+1$ vertices from $U_i$ need to be included in $S$, thus violating the membership constraint for the vertices in $V_i$. Observe that $S$ contains exactly one vertex from each set $V_i$, $i\in[k]$. As the $V_i$'s are a part of the clique $V'$ in $H$, if more than one vertex from a $V_i$ is included in the solution, then $|V'\cap S| > k$ and the membership constraint of vertices in $V'$ will be violated. Therefore exactly one vertex  from each  $V_i$ is included in $S$.
 		
\item $w\notin S$. \\
		The  vertex $w$ is part of the clique $V'$ and $S$ already contains $k$ vertices from $V'$ and any more vertices from $V'$ will violate the membership constraint of vertices in $V'$.
		
\item  $|U_i \cap S|= 0, i\in [1,k] $.\\
 		Every vertex in $U_i, i\in [1,k] $ is already dominated by a vertex in the corresponding $V_i$, and has membership value $1$. If a vertex from $U_i, i\in [1,k]$ is included in $S$, all vertices in $V_i$ will have membership $k+1$ leading to a  violation of the membership constraint for them. 
 
\item  $|S\cap D_{pq}| = 0, p,q\in[1,k], p < q$.\\
		Every vertex $x_{uv}\in D_{pq}$ is adjacent to $V_p\setminus u$ and $V_q\setminus v$, $u\in V_p, v\in V_q$. As $V'\cap S = k$, adding any vertex $x_{uv}\in D_{pq}$ to $S$ will violate the membership constraint of vertices in $V_p\setminus u$ and $V_q\setminus v$. Even if  $|V_p|=1$ or $|V_q|=1$,  adding $x_{uv}$ to $S$ will violate the membership constraint of $w\in V'$.
		%\textcolor{red}{Can we have a single vertex $w$ connected to everything in $D_{pq}$ do this purpose, instead of $k \choose 2 $ vertices?}.
\end{enumerate}
%\begin{observation}\label{obs:multicolored}
%$S = \{v_1, v_2, \ldots, v_k~|~ v_i\in V_i, 1\leq i \leq k \}$.
%\end{observation}
%\begin{proof}
  It follows from the above properties of $S$ that $S$ contains a vertex from each of the vertex sets $V_i$, $i\in[k]$. Let $S = \{v_1, v_2, \ldots, v_k~|~ v_i\in V_i, 1\leq i \leq k \}$. 
%\end{proof}
We now prove that the vertices corresponding to $S$ in $G$, say $S'$, form an independent set. Suppose not. This implies that $\exists v_i,v_j\in S$ such that $v_i v_j\in E(H)$. Without loss of generality, let $v_i\in V_i, v_j\in V_j$. Consider the vertex $x_{v_iv_j}\in D_{ij}$. Due to the construction of graph $H$, the vertex $x_{v_iv_j}$ is not adjacent to $v_i$ and $v_j$, and hence not dominated. This is a contradiction to the fact that $S$ is a  dominating set for $H$. Therefore, vertices $v_i,v_j \in S, i,j\in[1,k]$ cannot have an edge between them. Hence it is proved that $S'$ is a solution of size $k$ for the given instance of \mcis implying that $(G,k)$ is a YES instance.
\end{proof}
%\begin{mytheorem}[Proof of Theorem~\ref{thm:splitreduction}]
%{\em The \mmds problem is \textsc{W}[1]-hard for split graphs when parameterized by the size of membership.}
%\end{mytheorem}
\begin{proof} [Proof of Theorem~\ref{thm:splitreduction} ]
Lemma~\ref{lem:mcis-mmds} and Lemma~\ref{lem:mmds-mcis} along with the fact that \mcis is W[1]-hard \cite{FELLOWS200953} proves that \mmds is W[1]-hard parameterized by $k$, the membership.
\end{proof}

\section{Parameterizing \mmds by Vertex cover} 
\label{sec:parambyvc}

 First, we show that \mmds is FPT parameterized by vertex cover number, \textbf{vc}.
We then show that conditioned on the truth of the ETH,  \mmds does not have a subexponential algorithm in the size of vertex cover. 
 \subsection{\mmds is FPT parameterized by vertex cover}
In order to design an FPT algorithm parameterized by the size of a vertex cover of the input graph, we construct an FPT-time Turing reduction from \mmds to Integer Linear Programming (ILP, See Appendix for formal definition). In the reduced instance the number of constraints is at most twice the size of a minimum vertex cover. We then use the recent result by Dvořák \textit{et al.}~\cite{dvo-ijcai2017-85} which proves that ILP is FPT parameterized by the number of constraints. The following theorem directly follows from Corollary 9 of \cite{dvo-ijcai2017-85}.
%While reducing to ILP is a common approach in the design of parameterized algorithms, the reduction for \mmds is 

\begin{tcolorbox}
Integer Linear Programming\\
\begin{tabular}{llp{9cm}}
\textsf{Input}&:& A matrix $A \in \mathbb{Z}^{m \times \ell}$ and a vector $b \in \mathbb{Z}^m$.\\
{\sf Parameter}&: &$m$\\
{\sf Question}&:& Is there a vector $x \in \mathbb{Z}^{\ell}$ such that $A \cdot x \leq b$?
\end{tabular}
\end{tcolorbox}

\begin{theorem}[Corollary 9,~\cite{dvo-ijcai2017-85}]
\label{thm:ILPfptconsraints}
ILP is FPT  in the number of constraints and the maximum number of bits for one entry. 
\end{theorem}
\noindent
{\bf FPT time Turing reduction from \mmds to ILP}: Let $(G, k)$ be the input instance of \mmds.  Compute a minimum vertex cover of $G$, denoted by $C$, in time FPT in $|C|$ \cite{CyganFKLMPPS15}.
Let $I$ denote the maximum independent set $V \setminus C$.
%We start by computing a minimum vertex cover $C$ of $G$. It is well-known that a simple branching algorithm does this job in time $2^{|C|}\cdot n^{{\cal O}(1)}$. We simplify our task a bit by assuming that a minimum vertex cover $C$ of size $\ell$ has been computed in time $2^{\ell} \cdot n^{{\cal O}(1)}$. 
%Let $vc$ be the minimum size of the vertex cover of the input graph $G$.
The following lemma is crucial to the correctness of the reduction.

\begin{lemma}%$\diamondsuit$
\label{lem:kmdsvc}
 Let $D$ be a $k$ membership dominating set of $G$. Let $C_1 = D \cap C$,  $I_1 = I \setminus (N(C_1) \cap I)$, and 
 $R = N(C_1) \cap I \cap D$.  Then, $I_1 \subseteq D$, and $C \setminus (N[C_1] \cup N(I_1))$ is dominated by $R$.
\end{lemma}
\begin{proof}
 The outline is that $I_1$ cannot be dominated by any other vertex other than by itself. Further, $R \subseteq D$ is the the only vertices which can dominate $C \setminus (N[C_1] \cup N(I_1))$. Hence the lemma.
\end{proof}

As a consequence of this lemma, it is clear that the choice of $C_1$ immediately fixes $I_1$.  Thus, to compute the set $D$, the task is to compute $R$. We pose this problem as the {\em constrained} MMDS problem. A CMMDS problem instance is a 4-tuple $(G,k,C, C_1)$ where $C$ is a vertex cover and $C_1$ is a subset of $C$.  The decision question is whether there is a $k$ membership dominating set $D$ of $G$ such that $D \cap C = C_1$.
\begin{comment}
\\
\noindent\textbf{CMMDS}$(G =(V, E), k, C, C_1)$\\
Decide : $\exists S\subseteq B ~ \mid~ \forall a\in A, 1\leq |N[a]\cap S| + \lambda(a) \leq k$.

Minimize $M$ s.t,
\begin{align*}
       \lambda(a)+ \sum\limits_{b\in N(a)} y_b &\geq 1 \\
       \lambda(a)+ \sum\limits_{b\in N(a)} y_b &\leq M ~~~\forall a\in A,\\
        y_b\in \{0,1\}~~ \forall b \in B
 \end{align*}
 \end{comment}
From Lemma~\ref{lem:kmdsvc}, we know that given an instance of $(G,k,C,C_1)$, we know that $C_1$ immediately fixes $I_1 \subseteq I= V \setminus C$.  Thus, to compute $D$, we need to compute $R$ as defined in Lemma~\ref{lem:kmdsvc}. 
We now describe the ILP formulation to compute $R$ once $C_1$ (and thus $I_1$) is fixed.  Since $R$ is a subset of $I \setminus I_1$, it follows that the variables correspond to vertices in $I \setminus I_1$ which do not already have $k$ neighbors in $C_1$; we use $I_e$ to denote this set.  It can be immediately checked if $C_1 \cup I_1$ can be part of a feasible solution- we check that for no vertex is the intersection of its closed neighborhood greater than $k$.  We now assume that this is the case, and specify the linear constraints.  The linear constraints in the ILP are associated with the vertices in $C$. For each vertex in $C$ there are at most two constraints- if $v$ is in $C \setminus (N[C_1] \cup N(I_1))$, then at least one neighbor and at most $k$ neighbors from $I_e$ must be chosen into $R$.  On the other hand, for $v \in (N[C_1] \cap C) \cup N(I_1)$, we have the constraint that at most $k$ neighbors must be in $C_1 \cup I_1 \cup R$.  The choice of variables in $I_e$ does not affect any other vertex in $I$, and thus there are no costraints among the vertices in $I$.  
To avoid notation, we assume that an instance of CMMDS$(G,k,C,C_1)$, also denotes the ILP.  
%We now refer to this ILP as the tuple $(G,k,C,C_1)$ itself. 
\begin{lemma}%$\diamondsuit$
\label{lem:cmmdsfpt}
 The CMMDS problem on an instance $(G,k,C,C_1)$ can be solved in time which is FPT in the size of the vertex cover.
\end{lemma}
\begin{proof}
 Since the instance $(G,k,C,C_1)$ uniquely specifies the ILP for the choice of $R$, it follows that this ILP has ${\cal O}(|C|)$ constraints.
 From Theorem~\ref{thm:ILPfptconsraints}, we know that the ILP can be specified in FPT time with $|C|$ as the  parameter, and this proves the Lemma.
\end{proof}
\noindent
%We are now ready to complete our Turing reduction in FPT time. 
%In our FPT Turing reduction, we reduce an instance of \mmds with some additional constraint to an instance of an ILP that denotes the modified-constrained \mmds problem. We construct $2^{\ell}$ instances of ILP such that each instance of ILP has at most $\ell$ variables. 
%Moreover, at least one of the constructed instances of ILP is a Yes-instance if and only if the \mmds instance is a Yes-instance. 

%Now we consider the natural question whether there is an FPT algorithm for $k$- MDS for graphs of bounded vertex cover number. We answer this positively by a reduction to an integer linear program (ILP).
\noindent
%\begin{mytheorem}[Theorem~\ref{thm:vcFPT}]
%\textsc{MMDS} is FPT when parameterized by the size of vertex cover.
%\end{mytheorem}
\begin{proof} [Proof of Theorem~\ref{thm:vcFPT} ]
Given an input instance $(G,k)$ of \mmds, we first compute a minimum vertex cover $C$ in FPT time (in size of the cover as the parameter) using any of the well-known methods (see the book by Cygan et al.~\cite{CyganFKLMPPS15}, for example).   
%\textbf{Input}: $G = (V, E),~ k\in\mathbb{N}$.\\
%Assume $C\subseteq V$, $C$ is a vertex cover of $G$.\\
% 
%\noindent We are given the graph $G$ and the membership constraint $k$. A minimum vertex cover $C\subseteq V$ can be computed in time $2^{|C|} n^{{\cal O}(1)}$ by a simple branching algorithm\cite{CyganFKLMPPS15}. 
Let $I = V\setminus C$ be the independent set. Now we iterate through each subset $C_1$ of $C$, and check if it can be extended to a $k$ membership dominating set $D$ such that $D \cap C = C_1$.  For each such $C_1$, we know from Lemma~\ref{lem:kmdsvc} that $I_1 = I \setminus (N(C_1) \cap I)$ must be  added to the solution, if one exists.   For each subset $C_1\subseteq C$, we assume that $C_1 \cup I_1$ goes into the solution set.  Then, we solve CMMDS on the instance $(G, k, C, C_1)$  to check if there is an $R \subseteq I_e$ such that $C_1 \cup I_1 \cup R$ in a $k$ membership dominating set.  From Lemma~\ref{lem:cmmdsfpt}, we know that this check can be solved in FPT time. It thus follows that \mmds is FPT when parameterized by the size of vertex cover.
\end{proof}
\subsection{Lower bound assuming ETH}
%\begin{mytheorem}[Theorem~\ref{thm:vclowerbound}]
%{\em For a graph $G$ with vertex cover $VC(G)$, $|VC(G)|\geq 3$, there does not exist an algorithm running in time $2^{o(|VC(G)|)}poly(n) $ that solves \textsc{MMDS}, unless ETH fails.}
%\end{mytheorem}
We show that there is no sub-exponential-time parameterized algorithm for \mmds when the parameter is \emph{the vertex cover number}, using a reduction from $3$-\textsc{SAT}.  By the ETH, we know that $3$-\textsc{SAT} does not have a sub-exponential-time algorithm, and thus the reduction proves the lower bound for \mmds.
\begin{proof} [Proof of Theorem~\ref{thm:vclowerbound} ]
 Let $\phi$ be a boolean formula on $n$ variables $X = \{x_1, x_2,\dots, x_n\}$ having $m$ clauses $C = \{C_1, C_2,\dots , C_m\}$. We construct a graph $G=(V,E)$ from the input formula $\phi$ such that $\phi$ has a satisfying assignment if and only if $G$ has  a $k$ membership dominating set. \\
\textbf{Construction of graph $G$}\\
\begin{figure}[ht]
\centering
\includegraphics[scale=0.2]{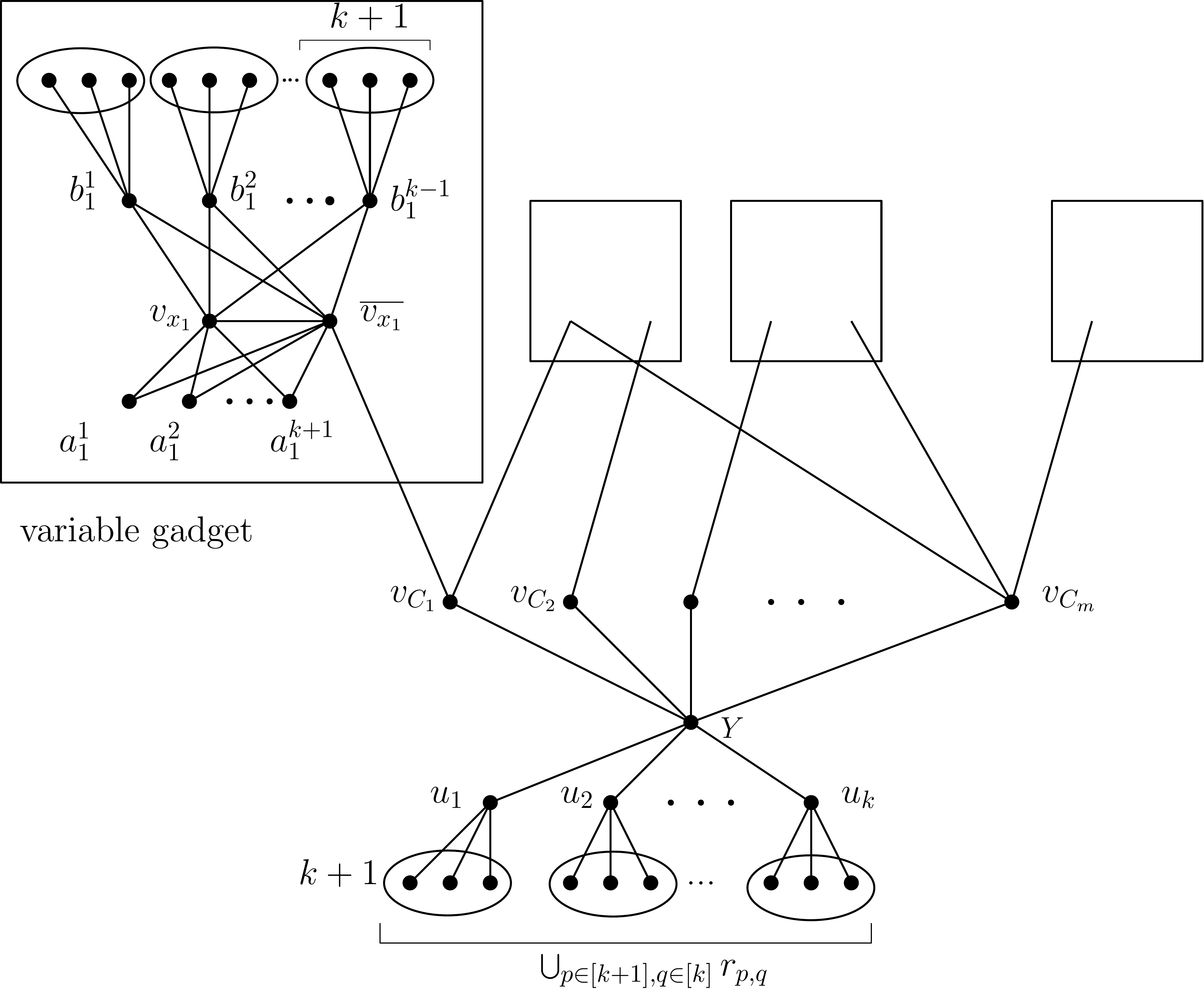} 
\caption{Construction for reduction from $3$-\textsc{SAT} to \mmds.} 
\label{fig:sat-to-kmmds}
\end{figure}
We construct a variable gadget and a clause gadget. For each variable $x_i,~1\leq i\leq n$ in $\phi$, create two vertices $v_{x_i}$ and $v_{\overline{x_i}}$, denoting its literals, with an edge between them. Make both $v_{x_i}$ and $v_{\overline{x_i}}$ adjacent to $k+1$ degree-two vertices labelled $a_i^j : 1\leq j \leq k+1$ and another set of $k-1$ vertices $b_i^j : 1\leq j \leq k-1$. Each $b_i^j$ is in turn adjacent to $k+1$ pendant vertices $d_{i,j}^t: 1\leq t \leq k+1$. This completes the variable gadget for a variable $x_i$.\\
\noindent For each  clause $C_l:1\leq l \leq m$, create a vertex $v_{C_l}$. For each clause $C_l$, make $v_{C_l}$ adjacent to a vertex $Y$. $Y$ is again connected to $k$ more vertices $u_q:1\leq q \leq k$. Each $u_q$ is in turn adjacent to $k+1$ pendant vertices $r_q^p : 1\leq p \leq k+1$. This is the clause gadget for graph $G$.\\
Finally, create  edges between clause vertices and those literal vertices which are in the clause.  The reduction is illustrated in
Figure~\ref{fig:sat-to-kmmds}.

\begin{claimnew}%$\diamondsuit$
\label{claim:gadgetVC}
The vertex cover number of graph $G$ is $(n+1)(k+1)$. 
\end{claimnew}
\begin{proof}
A minimum vertex cover of graph $G$  contains $\big\{v_{x_i}, v_{\overline{x_i}} ~|~ 1\leq i \leq n\big\}$, $\big\{b_i^j ~|~ 1\leq j \leq k-1, 1\leq i \leq n\big\}$, $Y$ and $\big\{u_q~|~1\leq q \leq k\big\}$ to cover all edges. Therefore $|VC(G)| = 2n + n(k-1) + k+1 = (n+1)(k+1)$. When $k$ is ${\cal O}(1)$, $|VC(G)| = {\cal O}(n)$.
\end{proof}

\begin{lemma}%$\diamondsuit$
\label{dSAT:MMDS}
If $\phi$ has a satisfying assignment then $G$ has a dominating set with membership value $k$.
\end{lemma}
\begin{proof}
Let $A:\{x_i | i\in [n]\}\to \{0,1\}$ be a satisfying assignment for $\phi$. Now, we construct a feasible solution $S$ for the MMDS problem as follows.
\begin{itemize}
	\item For each $i\in [n]$, 
	\begin{itemize}
		\item add  $v_{x_i}$ if $A[x_i] = 1$ or $v_{\overline{x_i}}$ if $A[x_i] = 0$  to $S$.
		\item add $\{b_i^j~ |~ 1\leq j \leq k-1\}$ to $S$.
	\end{itemize}
	\item Add $\{u_q~|~1\leq q \leq k\}$ to $S$.
	
\end{itemize}

We claim that $S$ is a dominating set for $G$ and has membership at most $k$. First, we show that $S$ is a dominating set. In the variable gadget, exactly one among $v_{x_i}$ or $v_{\overline{x_i}}$, and all  of $b_i^j$'s are in $S$. They dominate all other vertices in the variable gadget. It is given that $A$ is a satisfying assignment. i.e,   for each clause $C_l$, there is at least one literal assigned 1. Therefore each vertex $v_{C_l}$  is dominated by the vertex corresponding to the literal assigned 1 in $C_l$.  The vertices $Y$ and $\{r_q^p ~|~ p\in [k+1], q\in [k]\}$ are dominated by $\{u_q:q\in [k]\}$. \\
Next we show that the membership of any vertex in $G$ is at most $k$. In each variable gadget, maximum membership of $k$ is attained by the vertices $v_{x_i}$ and $v_{\overline{x_i}}$. Each clause vertex has membership at most 3. Vertex set $Y$ has the maximum membership of $k$ in clause gadget. 
\end{proof}

 \begin{lemma}
\label{MMDS:dSAT}
If $G$ has a dominating set with membership value $k$,  then  $\phi$ has a satisfying assignment. 
\end{lemma}
\begin{proof}
Let $S$ be feasible solution for MMDS in graph $G$. Then $S$ has the following properties:
\begin{itemize}
	\item In every variable gadget for a variable $x_i, ~1\leq i\leq n$,
	\begin{itemize}
		\item $\{b_i^j~ |~ 1\leq j \leq k-1\}$ must be there in $S$. If there is any $b_i^j\notin S$, all $d_{i,j}^t: 1\leq t \leq k+1$ should be included in $S$ which will violate the membership property by making the membership of $b_i^j$ to be $k+1$.
		\item Either $v_{x_i}$ or $v_{\overline{x_i}}$ must be there in $S$, in order to dominate $\{a_i^j ~|~ 1\leq j \leq k+1\}$. Note that both $v_{x_i}$ and  $v_{\overline{x_i}}$ together cannot be there in $S$ since it violates the membership property of both vertices.
	\end{itemize}
	\item $\{u_q:1\leq q \leq k\}\in S$. If any $u_q\notin S$, all $r_q^p : 1\leq p \leq k+1$ must be included in $S$ which violates the membership property for $u_q$. 
	\item $\{v_{C_l},~ 1\leq l \leq m\}\notin S$, since inclusion of any clause vertex $v_{C_l}$ violates the membership property for vertex $Y$.
	\item Out of the three literal vertices in any clause $C_l$, atleast one will be included in  $S$ in order to dominate corresponding clause vertex $v_{C_l}$.
\end{itemize}
  
    It follows from the above properties that atleast one literal vertex from every clause will be included in $S$ and assigning 1 to those literals makes a satisfying assignment for the boolean formula $\phi$.
\end{proof}
From Lemma~\ref{dSAT:MMDS} and Lemma~\ref{MMDS:dSAT}, it follows that the $3$-\textsc{SAT} can be reduced to \mmds parameterized by vertex cover number. Therefore a $2^{o(\textbf{vc}(G))}n^{{\cal O}(1)}$ algorithm for \mmds will give a $2^{o(n)}$ algorithm for $3$-\textsc{SAT} which is a violation of ETH. Hence it is proved that there is no sub-exponential algorithm for \mmds when parameterized by vertex cover number.
\end{proof}

\section{Conclusion}

In this paper we study the parameterized complexity of the \textsc{Minimum Membership Dominating Set} problem,  which requires finding a dominating set such that each vertex in the graph is dominated minimum possible times. We start our analysis by showing that in spite of having no constraints on the size of the solution, unlike \textsc{Dominating Set}, MMDS turns out to be W[1]-hard when parameterized by pathwidth (and hence treewidth). We further show that the problem remains W[1]-hard for split graphs when the parameter is the size of the membership. For general graphs we  prove that MMDS is FPT when parameterized by the size of vertex cover. Finally, we show that assuming ETH, the problem does not admit a sub-exponential algorithm when parameterized by the size of vertex cover, thus showing our FPT algorithm to be optimal. There are many related open problems that are yet to be explored. One such problem is analyzing the complexity of \mmds in chordal graphs. Other directions involve structural parameterization of \mmds with respect to other parameters such as maximum degree, distance to bounded degree graphs, bounded genus and maximum number of leaves in a spanning tree.
We have got an idea from an anonymous reviewer from IPEC 2021 that a W[2] hardness could be proved. We are working on it.

%%
%% Bibliography
%%

%% Please use bibtex, 
\bibliographystyle{splncs04}
\bibliography{references}
%\pagebreak

%\input{input/appendix.tex}

\end{document}